\documentclass[12pt, onecolumn]{IEEEtran}
\usepackage[utf8]{inputenc} 
\usepackage[T1]{fontenc}
\usepackage{url}              % provides \url{...}
\usepackage{cite}             % improves presentation of citations

\usepackage[cmex10]{amsmath}  % Use the [cmex10] option to ensure complicance
\usepackage{amssymb}
\usepackage{braket}

\usepackage{algorithm}
\usepackage{algpseudocode}

\usepackage{subcaption}

 \algnewcommand{\IfThenElse}[3]{% \IfThenElse{<if>}{<then>}{<else>}
 	\State \algorithmicif\ #1\ \algorithmicthen\ #2\ \algorithmicelse\ #3}
  
\usepackage{comment}
\usepackage{bm}% bold math

\usepackage{tikz}	
\usepackage{xcolor}
\usetikzlibrary{backgrounds,fit,decorations.pathreplacing}  % TikZ libraries
\usetikzlibrary{automata} % LATEX and plain TEX
\usetikzlibrary{circuits.logic.CDH}
\usetikzlibrary{circuits.logic.US}
\usetikzlibrary{mindmap} 
\usetikzlibrary{trees}
\usetikzlibrary{shapes,decorations.pathreplacing}
\usetikzlibrary{arrows.meta}
\usetikzlibrary{shapes.geometric}

\DeclareMathOperator*{\argmax}{arg\,max}

\newcommand{\OSD}{\operatorname{OSD}}

 \usepackage{xspace}
\newcommand{\BPOSD}{MBP$_4$+OSD$_4$}

\usepackage{ulem}
\usepackage{enumitem}
\usepackage{mleftright}       % fix to wrong spacing of \left-,
\mleftright                   % \middle- \right-commands 

\usepackage{graphicx}         % provides \includegraphics{...} to
\usepackage[space]{grffile}			% KY: support file name with spaces 
\usepackage{booktabs}         % fixes poor spacing in tables and
\let\emph\textit
\usepackage{amsthm}
\theoremstyle{definition}		% adds extra space above and below, but sets the text in roman (words like Definition be bold)
\newtheorem{definitionenv}{Definition}
\newtheorem{lemmaenv}[definitionenv]{Lemma}
\newtheorem{theoremenv}[definitionenv]{Theorem}
\newtheorem{corollaryenv}[definitionenv]{Corollary}
\newtheorem{propositionenv}[definitionenv]{Proposition}
\newtheorem{conjectureenv}[definitionenv]{Conjecture}
\newtheorem{remarkenv}[definitionenv]{Remark}
\newenvironment{remark}{\begin{remarkenv}\rm}{\end{remarkenv}}
\newcommand{\br}{\begin{remark}}
	\newcommand{\er}{\end{remark}}

\newtheorem{exampleenv}{Example}
\newtheorem{app-lemmaenv}[section]{Lemma}

\newenvironment{definition}{\begin{definitionenv}\rm}{\end{definitionenv}}
\newenvironment{lemma}{\begin{lemmaenv}\rm}{\end{lemmaenv}}
\newenvironment{theorem}{\begin{theoremenv}\rm}{\end{theoremenv}}
\newenvironment{corollary}{\begin{corollaryenv}\rm}{\end{corollaryenv}}
\newenvironment{example}{\begin{exampleenv}\rm}{\end{exampleenv}}
\newenvironment{proposition}{\begin{propositionenv}\rm}{\end{propositionenv}}
\newenvironment{conjecture}{\begin{conjectureenv}\rm}{\end{conjectureenv}}
\newenvironment{app-lemma}{\begin{app-lemmaenv}\rm}{\end{app-lemmaenv}}

\newcommand{\bd}{\begin{definition}}
	\newcommand{\ed}{\end{definition}}
\newcommand{\bl}{\begin{lemma}}
	\newcommand{\el}{\end{lemma}}
\newcommand{\elp}{\hspace*{\fill} $\Box$
\end{lemma}}
\newcommand{\bt}{\begin{theorem}}
\newcommand{\et}{\end{theorem}}
\newcommand{\etp}{\hspace*{\fill} $\Box$
\end{theorem}}
\newcommand{\bc}{\begin{corollary}}
\newcommand{\ec}{\end{corollary}}
\newcommand{\ecp}{\hspace*{\fill} $\Box$
\end{corollary}}
\newcommand{\bcj}{\begin{conjecture}}
\newcommand{\ecj}{\end{conjecture}}

\newcommand{\be}{\begin{example}}
\newcommand{\ee}{\end{example}}
\newcommand{\eep}{\hspace*{\fill} $\Box$
\end{example}}
\newcommand{\bp}{\begin{proposition}}
\newcommand{\ep}{\end{proposition}}
\newcommand{\epp}{%\hspace*{\fill} $\Box$
\end{proposition}}

\newcommand{\cC}{{\cal C}}

\newcommand{\cG}{{\cal G}}

\newcommand{\cS}{{\cal S}}

\newcommand{\mbp}{{\mathbf{p}}}
\newcommand{\mbq}{{\mathbf{q}}}

\newcommand{\mbg}{{\mathbf{g}}}
\newcommand{\mbh}{{\mathbf{h}}}
\newcommand{\mbe}{{\mathbf{e}}}

\newcommand{\mbs}{{\mathbf{s}}}

\newcommand{\mbA}{{\mathbf{A}}}
\newcommand{\mbB}{{\mathbf{B}}}
\newcommand{\mbC}{{\mathbf{C}}}

\newcommand{\mbI}{{\mathbf{I}}}

\newcommand{\mbH}{{\mathbf{H}}}

\newcommand{\mbL}{{\mathbf{L}}}

\newcommand{\Tau}{{\mathrm{T}}}
 
\newcommand{\wt}[1]{\mathrm{wt}\left(#1\right)}

\usepackage[colorinlistoftodos,prependcaption]{todonotes}
\usepackage{xcolor}
\usepackage{xargs}

\newcommandx{\yellownote}[2][1=]{\todo[inline,linecolor=yellow,backgroundcolor=yellow!25,bordercolor=yellow,#1]{#2}}

\begin{document}

\title{Efficient Approximate Degenerate Ordered Statistics Decoding for Quantum Codes via Reliable Subset Reduction}
\author{Ching-Feng Kung,\, Kao-Yueh~Kuo, \,and\, Ching-Yi~Lai
\thanks{\footnotesize
This article was presented in part at the 2023 International Symposium on Topics in Coding (ISTC)~\cite{KKL23}.

	CYL   was supported by the National Science and Technology Council  in Taiwan, under Grant Nos. 113-2221-E-A49-114-MY3 and 113-2119-M-A49-008-.
 
CFK and CYL are with the Institute of Communications Engineering, National Yang Ming Chiao Tung University, Hsinchu 300093, Taiwan. (email: cylai@nycu.edu.tw)

KYK is with the School of Mathematical and Physical Sciences, University of Sheffield, UK.
}
}
\maketitle

\begin{abstract}

	   Efficient and scalable decoding of quantum codes is essential for high-performance quantum error correction. 
	  	{ In this work, we introduce Reliable Subset Reduction (RSR), a reliability-driven preprocessing framework that leverages belief propagation (BP) statistics to identify and remove highly reliable qubits, substantially reducing the effective problem size. }
	   Additionally, we identify a degeneracy condition that allows high-order OSD to be simplified to order-0 OSD. By integrating these techniques, we present an ADOSD algorithm that significantly improves OSD efficiency.
 	   {Our BP+RSR+ADOSD framework extends naturally to circuit-level noise and can handle large-scale codes with more than $10^4$ error variables. }
 	   {
Through extensive simulations, we demonstrate improved performance over MWPM and Localized Statistics Decoding for a variety of CSS and non-CSS codes under the code-capacity noise model, and for rotated surface codes under realistic circuit-level noise.
	     At low physical error rates, RSR reduces the effective problem size to as little as 1\% (e.g., for $\epsilon=0.001$ in surface-code DEM), enabling higher-order OSD with  drastically reduced computational complexity. These results highlight the practical efficiency and broad applicability of the BP+ADOSD framework for both theoretical and realistic quantum error correction scenarios.
	}

\end{abstract}

\section{Introduction}
\label{sec:intro}

Reliable quantum communication is a critical area of research, essential for scaling quantum systems, enabling the exchange of quantum information across distances, and facilitating multiparty protocols. Quantum states are inherently fragile, requiring the implementation of quantum error correction to mitigate the effects of noise and decoherence~\cite{Got14,Ter15,LK24}.

Quantum stabilizer codes, analogous to classical linear block codes, allow efficient encoding and binary syndrome decoding~\cite{GotPhD,CRSS98,NC00}. A notable class of stabilizer codes is quantum low-density parity-check (LDPC) codes, which are preferred for their high code rates and feasible syndrome measurements from low-weight stabilizers. These codes can be decoded using belief propagation (BP), similar to classical LDPC codes~\cite{MMM04,PC08,Wan+12,Bab+15,ROJ19,KL20,KL21a,Gal62,MN95,Mac99,Tan81,Pea88,KFL01}.
A quantum LDPC code of the Calderbank--Shor--Steane (CSS) type \cite{CS96,Ste96} can be decoded by the BP algorithm for binary codes (BP$_2$) 
by treating $X$ and $Z$ errors separately.
  When $X$ and $Z$ errors exhibit correlations, such as in depolarizing channels,  a quaternary BP algorithm (BP$_4$) could be more effective, as it better exploits the error correlations during decoding. An important feature of quantum codes is their binary error syndromes, unlike classical nonbinary codes. 
  Consequently, the complexity of BP$_4$ can be reduced by restricting message passing in the Tanner graph~\cite{KL20} to scalar messages, which encode the likelihood ratios associated with Pauli commutation relations.

The decoding problem for highly degenerate quantum codes, such as topological codes~\cite{Kit03,BM06,BM07,HFDM12,THD12,LAR11,ATBFB21}, remains particularly challenging for belief propagation (BP). A variety of remedies have been proposed to improve BP convergence and accuracy, including random perturbation~\cite{PC08}, enhanced feedback~\cite{Wan+12,Bab+15}, check-matrix augmentation~\cite{ROJ19}, message normalization and offsets~\cite{KL20,KL21a}, and trapping-set analysis~\cite{RV21,CPR+24}. A notable advance is MBP$_4$~\cite{KL22}, which introduces memory effects and enables BP to decode many topological codes.

When BP alone is insufficient, postprocessing decoders such as ordered statistic decoding (OSD) are commonly employed. OSD was originally developed in classical coding theory as a controllable-complexity approximation to maximum-likelihood decoding~\cite{FL95,Wolf78,LV95,Eli57,Eli91,For66,Cha72,SB89,KNIH94,HHC93}, and later adapted to syndrome-based and quantum decoding~\cite{FLS98,PK21}. Its core operation is Gaussian elimination on an $m\times n$ parity-check matrix, followed by flipping unreliable bits, but high-order OSD incurs significant computational cost~\cite{FL96,GS97,AL12,YSLV19,YSVL21,YCC22}.

In the quantum regime, several BP–OSD hybrids and variants have been proposed~\cite{RWBC20,iOM24closed,GCR24}, along with divide-and-conquer approaches such as ambiguity clustering (AC)~\cite{WB24ambiguity} and localized statistics decoding (LSD)~\cite{hillmann2025localized}. These methods attempt to reduce the Gaussian elimination burden by partitioning the problem.
%, but OSD-0 or its variants still dominate the overall complexity.

Despite these developments, most existing postprocessing schemes under the code-capacity noise model rely on binary BP applied to quaternary quantum errors, thereby discarding $X/Z$ correlations and leading to intrinsic suboptimality. In contrast, the quaternary AMBP$_4$ decoder~\cite{KL22} directly exploits these correlations and has been shown to outperform all of the above decoders on a wide range of quantum codes, including topological codes~\cite{BM07}, BB codes~\cite{BCG+24}, and generalized hypergraph-product codes~\cite{PK21,RWBC20}.
{However, AMBP$_4$ still fails on certain important code families, such as lift-connected surface (LCS) codes~\cite{ORM24}, motivating the need for a more powerful postprocessing framework.}

   {
   In this paper, we investigate postprocessing techniques for enhancing the performance of BP. When BP fails to converge to a valid global solution, it still provides rich statistical information about the error configuration. In particular, observing the convergence behavior of belief propagation on sparse graphs, most error variables are typically resolved with high confidence, except for a small number of unsettled variables associated with trapping sets~\cite{RV21}. Consequently, a large fraction of variables can be regarded as already determined and effectively part of the solution.
}
   
   {Based on this observation, we propose reliable subset reduction (RSR). In the code-capacity decoding problem of an $[[n,k]]$ quantum code, decoding amounts to solving a linear system with $2n$ binary variables subject to $n-k$ independent constraints. The central idea of RSR is that if a subset of variables can be identified as highly reliable, the linear system can be reduced by fixing those variables, thereby substantially lowering the computational burden of subsequent postprocessing steps, such as Gaussian elimination.
   
   To identify such a reliable subset, we exploit both the hard-decision history and the final soft information produced by BP$_4$~\cite{KKL23}.} Variables whose hard decisions remain nearly constant across BP iterations and whose final beliefs have high confidence are classified as reliable. At low physical error rates, a large fraction of variables satisfy these criteria, so eliminating them dramatically reduces the effective problem size and accelerates downstream postprocessing decoders.

   {We apply OSD after RSR, using the reliabilities produced by RSR to order the remaining unreliable variables.} The resulting OSD variant, which is driven by BP$_4$ statistics, is denoted as OSD$_4$.

We next consider approximate degenerate OSD for quantum codes. In typical OSD, a list of error candidates consistent with a given syndrome is generated and the minimum-weight candidate is selected. For quantum codes, however, the optimal criterion is to choose the most probable error coset, which is generally intractable. To address this, we adopt an approximate degenerate decoding strategy in which only low-weight representatives within each coset are considered.
{When RSR is applied, the resulting reduced problem makes such approximate degenerate decoding feasible, in a manner analogous to the ambiguity clustering (AC) approach~\cite{WB24ambiguity}.} For computational efficiency, we nevertheless use minimum-weight selection among the retained candidates in our simulations.

Additionally, we introduce a degeneracy-aware pruning condition for stabilizer codes. We show that certain OSD bit flips correspond to multiplying stabilizers and therefore generate degenerate error candidates: although their weights may differ, they belong to the same error coset. Such flips can be ignored in approximate degenerate OSD.
{ A key observation is that when all candidate flips correspond to the trivial logical coset (i.e., zero logical syndrome), high-order OSD is unnecessary.}

By combining these ideas, we obtain an approximate degenerate OSD algorithm, which we denote as ADOSD$_4$.

	Finally, there has been substantial recent progress in circuit-level decoding~\cite{Pry20,du2022stabilizer,du2024check,BCG+24,SMR+24}. A variety of decoders have been developed for circuit-level noise models, including BP-OSD-CS, GDG, AC~\cite{WB24ambiguity}, LSD~\cite{hillmann2025localized}, and ordered Tanner forest~\cite{iOMR+24}.
	Our BP–OSD framework can also be extended to handle both data and syndrome errors under the phenomenological noise model~\cite{QVRCT21,HB22}, using techniques proposed in~\cite{ALB20,KCL21,KL25}, as demonstrated in our recent work~\cite{KL24b}.

{
	In this paper, we further show that our approach can be directly applied to the full circuit-level noise model through the detector error model (DEM) generated by STIM~\cite{Gid21stim}. The resulting binary DEM decoding problem can be handled using the proposed BP–RSR–ADOSD procedure.
}

We conduct simulations of  MBP$_4$ and our OSD schemes OSD$_4$ and ADOSD$_4$ on several quantum codes under depolarizing errors in the code capacity noise model, including   BB codes~\cite{BCG+24}, rotated toric codes and rotated surface codes~\cite{BM07}, (6.6.6) and (4.8.8) color codes~\cite{LAR11}, twisted XZZX codes on a torus~\cite{KDP11},   GHP codes~\cite{PK21,RWBC20},
 and  LCS codes~\cite{ORM24}.
The results demonstrate that our proposed schemes outperform previous approaches in the literature, both in terms of error thresholds for topological codes and low error rate performance (logical error rate around $10^{-6}$).

{
	In DEM circuit-level setting for rotated surface codes, we further show that MBP$_4$+ADOSD$_4$ consistently outperforms LSD. When both $X$- and $Z$-type detector variables are used, MBP$_4$+ADOSD$_4$ also surpasses the minimum-weight perfect-matching (MWPM) decoder~\cite{Edm65,WFSH10} and achieves an improved threshold of approximate $0.76\%$.
}

{
 Beyond improved decoding performance, a central outcome of our simulations is the strong reduction effect produced by reliable subset reduction (RSR): at low physical error rates, RSR typically reduces the effective decoding problem size by one to two orders of magnitude (e.g., below $4\%$ of the original size for $\epsilon\le 0.005$, and below $1\%$ for $\epsilon=0.001$ in surface-code DEM problems).
 This dramatic dimensionality reduction enables higher-order OSD to be applied efficiently even for problems with more than $10^4$ error variables, a regime that is far beyond the reach of conventional OSD-based postprocessing.
 As a result, our MBP–RSR–ADOSD approach achieves strong threshold performance for topological codes and substantially improved logical error rates in the low-error regime.
}
 {
 Crucially, these results are achieved under a strict budget of at most 10 BP iterations. In practice, when MBP$_4$ converges, it typically does so within 3 to 6 iterations.
}

This paper is organized as follows. We introduce stabilizer codes in Section~\ref{sec:stb}.
The reliable subset reduction method and the associated reliability metrics are presented in Section~\ref{sec:RSR}.
We then introduce the associated OSD algorithm in Section~\ref{sec:BPOSD}.
Approximate degenerate decoding is covered in Section~\ref{sec:ADOSD}.
The MBP-ADOSD procedure is extended to circuit-level decoding via DEM in Section~\ref{sec:DEM}, and simulation results are provided in Section~\ref{sec:sim}.
Related work is discussed in Section~\ref{sec:related}.
Finally, we conclude in Section~\ref{sec:conclusion}.

 \section{Quantum stabilizer codes} \label{sec:stb}
In this section, we review the basic of stabilizer codes and define the relevant notation \cite{GotPhD,CRSS98,NC00}.  

Let $I=\begin{bmatrix}
	1&0\\0&1
\end{bmatrix},$ $ X=\begin{bmatrix}
	0&1\\1&0
\end{bmatrix},$ $ Z=\begin{bmatrix}
	1&0\\0&-1
\end{bmatrix},$ and $ Y=iXZ$ denote the Pauli matrices. 
 The $n$-fold Pauli group, $\mathcal{G}_n$,  is the group of  $n$-qubit Pauli operators, defined as
$$\mathcal{G}_n \triangleq \left\{ \omega M_1\otimes\dots\otimes M_{n}: \omega\in \{\pm 1, \pm i\}, M_j\in \{I, X, Y, Z \} \right\}.$$ 
Every nonidentity Pauli operator in $\mathcal{G}_n$ has eigenvalues $\pm 1$. Any two Pauli operators  either commute or anticommute with each other. 
The weight of a Pauli operator $g\in\cG_n$, denoted $\wt{g}$, refers to the number of its nonidentity components.

We consider Pauli errors, assuming independent depolarizing errors with rate $\epsilon$. Each qubit independently experiences an $X$, $Y$, or $Z$ error with probability $\epsilon/3$ and no error with probability $1-\epsilon$.
Therefore, at a small error rate $\epsilon$, low-weight Pauli errors are more likely to occur across $n$ qubits. 

A \emph{stabilizer group} $\mathcal{S}$ is an Abelian subgroup in $\mathcal{G}_n$ such that $-I^{\otimes n}\not\in \mathcal{S}$ \cite{NC00}. 
Suppose that 
$\mathcal{S}$ is generated by $n-k$ independent generators.
Then $\mathcal{S}$ defines an  $[[n,k,d]]$ stabilizer code $\mathcal{C}(\mathcal{S})$ that encodes $k$ logical qubits into $n$ physical qubits:	
$$\mathcal{C}(\mathcal{S}) = \left\{ \ket{\psi}\in  \mathbb{C}^{2^n} : g\ket{\psi} = \ket{\psi}, ~\forall\, g \in \mathcal{S} \right\}.$$
The elements in $\mathcal{S}$ are called \emph{stabilizers}.
 The  parameter $d$ represents the minimum distance of the code such that any Pauli error of weight less than $d$ is \textit{detectable}.

An error $e\in\cG_n$ can be detected by $\cC(\cS)$ through stabilizer measurements  if $e$ anticommutes with some of the stabilizers. 
Suppose that $\{g_i\}_{i=1}^m$, where $m \ge n-k$, is  a set of stabilizers that generates $\cS$. Their binary measurement outcomes  are referred to as the \emph{error syndrome} of $e$.
Let $N(\cS) \subset \cG_n$ denote the normalizer group of $\cS$, which consists of the Pauli operators that commute with all stabilizers. Consequently, if an error is in $N(\cS)$, it will have zero syndrome and cannot be detected. Note that if an error is a stabilizer, it has no effect on the code space.
An element in   the normalizer group $N(\cS)$ that is not a stabilizer, up to a phase, is called a  \textit{nontrivial logical operator as it changes the logical state of a code} .  Therefore, the minimum distance of $\cC(\cS)$ is defined as the minimum weight of a nontrivial logical operator.
It can be observed that  $e\in\cG_n$ and $eg$ for $g\in\cS$ have the same effects on the codespace. They are  called \textit{degenerate} to each other.
We say that a quantum code is \textit{highly degenerate} if it has many stabilizers with low weight relative to its minimum distance.

{
A syndrome decoding problem is defined as follows: given a measured syndrome corresponding to an unknown Pauli error $e$, a \emph{decoder} outputs an estimate $\hat e \in \mathcal G_n$.
An error estimate $\hat{e}$ is considered \emph{valid} if it reproduces the measured syndrome.
	The decoder is said to \emph{succeed} if the estimate is degenerate with the true error, i.e.,
\[
\hat{e} e \in \cS,
\]
or equivalently, $\hat{e} e$ commutes with all elements of the normalizer group $N(\cS)$.
Otherwise, the decoder produces a logical error.}

Two decoding criteria are typically considered: minimum weight decoding and degenerate maximum likelihood decoding \cite{HL11,KL13_20,IP15}. In minimum weight decoding, the error $\hat{e}$ with the smallest weight is chosen, while in degenerate maximum likelihood decoding, the error $\hat{e}$ whose coset $\hat{e} \mathcal{S}$ has the minimum coset probability is selected. However, degenerate maximum likelihood decoding has exponential complexity and is impractical due to the need for coset enumeration. We propose the following criterion.

\begin{definition} (Approximate Degenerate Decoding)\\
 A $\delta$-approximate degenerate decoding ($\delta$-ADD) criterion  aims to find an error estimate  $\hat{e}$ that matches the given syndrome, maximizing the dominant terms of its coset probability 
 \begin{align}
     \sum_{g\in\cS, \wt{g}\leq \delta }\Pr \big\{ \hat{e} g\big\}.
 \end{align}
\end{definition}
Note that $0$-ADD reduces to minimum weight decoding, while $n$-ADD corresponds precisely to degenerate maximum likelihood decoding.

\subsection{Binary representations}

We can study the decoding problem in the binary vector space~\cite{GotPhD,NC00}, using a mapping $\varphi: \cG_1\rightarrow \{0,1\}^2$ 
    $$I \mapsto  \left[\begin{array}{c|c}
        0 &  0
        \end{array}\right],X\mapsto \left[\begin{array}{c|c}
        1 &  0
        \end{array}\right], Z\mapsto \left[\begin{array}{c|c}
        0 &  1
        \end{array}\right], Y\mapsto \left[\begin{array}{c|c}
        1 &  1
        \end{array}\right].$$
We extend the $\varphi$ to $\cG_n$ as follows:  For $g=M_{1} \otimes M_{2} \otimes \dots\otimes M_{n}\in\cG_n$, 
\[
\varphi(g)\triangleq \mbg	= \left[\begin{array}{c|c}
	\mbg^X &  \mbg^Z
\end{array}\right] = \left[\begin{array}{c|c}
a_1 \cdots  a_n &  b_1 \cdots b_n 
\end{array}\right] \in\{0,1\}^{2n},
\]
where $\mbg^X=(a_1, \dots, a_n)\in\{0,1\}^n$ and $\mbg^Z=(b_1,\dots,b_n)\in \{0,1\}^n$ are the indicator vectors of $X$ and $Z$ components of $g$, respectively,
such that  $\varphi(M_{j})=\left[
\begin{array}{c|c}
	a_j &  b_j
\end{array}\right].$

A \emph{check matrix} of  $\cC(\cS)$ corresponding to a set of stabilizer generators $\{g_i\}_{i=1}^m$
is an $m\times 2n$ binary matrix $\mbH$  with rows $\varphi(g_1),\dots, \varphi(g_m)$.

 An error $e\in\cG_n$ can be represented by $\mbe=\varphi(e)  \in\{0,1\}^{2n}$ and
 its error syndrome is 
 \begin{align}
 \mbs = \mbH \Lambda \mbe^\top 
   \in\{0,1\}^{m\times 1}, \label{eq:syndrome}
 \end{align}
 where  $\mbe^\top$ denotes the transpose of $\mbe$,
 $\Lambda = \left[ \begin{array}{c|c}\mathbf{0} & \mbI_n\\ \mbI_n & \mathbf{0} \end{array}\right],$ 
 $\mbI_n$ is the $n\times n$ identity matrix, 
 and  $\mathbf{0}$ is the zero matrix of appropriate dimensions.

Suppose this code has $2k$ independent logical operators, $\bar{X}_i$ and $\bar{Z}_j$ for $i,j=1,\dots,k$,  which, together with the stabilizers $\{g_i\}$, generate the normalizer group $N(\cS)$. The $2k\times 2n$ binary logical matrix $\mbL$ is defined with rows $\varphi(\bar{X}_1),\dots, \varphi(\bar{X}_k), \varphi(\bar{Z}_1),\dots, \varphi(\bar{Z}_k)$. 
Since stabilizers commute with logical operators, we have $\mbH\Lambda \mbL^\top=\mathbf{0}$.

The decoding problem is to solve the   system of linear equations (\ref{eq:syndrome}) with  $2n$ binary variables in $\mbe\in\{0,1\}^{2n}$ that are most probable.
Note that $ \mbH$ is of rank $n-k$ so
the degree of freedom of this binary system is $n+k$.

{
An error $e \in \cG_n$ is classified as a \emph{logical error} for a decoder if the decoder fails to produce a valid estimate, or outputs an estimate $\hat{e} \in \cG_n$ such that
\begin{align}
	\mbH \Lambda \mbe^\top = \mbH \Lambda \hat{\mbe}^\top
\end{align}
but
\begin{align}
	\mbL \Lambda \mbe^\top \neq \mbL \Lambda \hat{\mbe}^\top .
\end{align}
}
{
	We define the \emph{logical syndrome} of an error $\mbe$ as
	\begin{align}
		\mbL \Lambda \mbe^\top \in \{0,1\}^{2k \times 1}. \label{eq:logical_syndrome}
	\end{align}
}

\subsection{Belief Propagation (BP) Decoding} \label{sec:BPs}
We  explain how a quaternary BP algorithm works in the following~\cite{KL20}, which is computed in linear domain. 
A log-likelihood version of BP can be found in~\cite{KL21a}.

 {Given a check matrix $\mbH$, a syndrome $\mbs$, and the prior error distributions 
 	$\mbp_{i}\in\mathbb{R}^4$ for qubit $i\in\{1,2,\dots,n\}$,}
 BP performs the following steps:
 \begin{enumerate}
 	\item The belief distribution vector $\mbq_{i} = (q_i^I, q_i^X, q_i^Y, q_i^Z)$ is updated by BP, using the parity-check messages and the prior $\mbp_{i}$, for $i=1,2,\ldots,n$.
 	
 	\item At each iteration, an error estimate $\hat{\mbe}\in\{0,1\}^{2n}$ is formed by hard decisions on $\mbq_{i}$. 
 	If $\hat{\mbe}$ is \emph{valid}, i.e., it matches the measured syndrome, it is output as the BP estimate; otherwise, proceed to the next iteration starting from step~1).
 \end{enumerate}
 These two steps are iterated for at most $\Tau$ iterations, where $\Tau$ is chosen in advance. 
 If no valid estimate is obtained after $\Tau$ iterations, BP declares a decoding failure.

{
	\section{The Reliable Subset Reduction Method} \label{sec:RSR}
	
	In practice, a BP decoder typically either converges within a small number of iterations, on the order of $O(\log\log n)$ on average, or remains unsettled until the maximum number of iterations is reached.  
	In the latter case, although BP fails to converge globally, most error variables have already been decided with high confidence.  
	From the BP output, a large fraction of variables are therefore very reliable and can be regarded as part of the solution.  
	The remaining unsettled variables are usually associated with trapping sets~\cite{RV21}, and additional post-processing techniques are required to resolve them.

}

 \subsection{Reliable subset reduction for the syndrome decoding problem}\label{sec:rsr_decoding}
 
 In the syndrome decoding problem, we seek a binary vector $\mbe\in\{0,1\}^{2n}$ satisfying
 \[
 \mbH\Lambda \mbe^\top = \mbs,
 \]
 where $\mbH$ has rank $n-k$ and hence the solution space has dimension $n+k$. 
 When BP fails to converge to a valid solution within $\Tau$ iterations, a postprocessing algorithm such as OSD is typically required.
 
  {
We expect that after a few BP iterations, most variables have already converged to highly reliable values, while only a small subset remains ambiguous.
 We denote by $R\subset\{e_1,\dots,e_{2n}\}$ a subset of \emph{reliable} variables, whose values are assumed to be correct and fixed.
 The remaining variables are treated as unknowns.
}

 By eliminating the reliable variables from the parity-check equations, the original syndrome decoding problem can be reduced to a much smaller linear system.
 We refer to this procedure as \emph{reliable subset reduction (RSR)}.
 The number of remaining variables after reduction is called the \emph{effective length}.

 \begin{algorithm}[ht] \caption{RSR} \label{alg:RSR}
 	
 	\textbf{Input}: parity-check matrix $\mbH\in\{0,1\}^{m\times 2n}$, syndrome $\mbs$, reliable subset $R$.
 	
 	\textbf{Output}: reduced matrix $\tilde{\mbH}$, reduced syndrome $\tilde{\mbs}$, and permutation $\sigma$.
 	
 	{\bf Steps}: 
 	\begin{algorithmic}[1] \itemsep=6pt

 		\State Let $\sigma$ be a column permutation such that
 		$\sigma(\mbe) =  \begin{bmatrix}
 			\mbe'&\mbe^{R}
 		\end{bmatrix}$,
 		where   $\mbe^{R}$ is the reliable subset of $v$ bits,
 		and    $\mbe'$ is the remaining $2n-v$ bits. 
 		Let $\sigma (\hat{\mbe})= 
 		\begin{bmatrix}
 			\hat{\mbe}'&\mbe^{R}
 		\end{bmatrix}$ be a permuted error estimate,
 		which  contains
 		the   reliable bits in  $\mbe^{R}$. 
 		
 		\State Let $\tau$ be a row permutation such that $
 		\tau(\sigma(\mbH\Lambda))= \begin{bmatrix}
 			\Tilde{\mbH} &  \mbB\\
 			\mathbf{0} &\mbC\\
 		\end{bmatrix}$ for $\tilde{\mbH}\in\{0,1\}^{m'\times (2n-v)}$ with $m'\leq m$. The $m-m'$ rows  $\begin{bmatrix}
 			\mathbf{0} &\mbC
 		\end{bmatrix}$  correspond to  parity checks related only to the  reliable bits.  
 		Assume that $\tau (\mbs)= \begin{bmatrix}
 			\mbs'' \\ \mbs^{R}
 		\end{bmatrix}$, where $\mbs''\in\{0,1\}^{m'\times 1}$,
 		and $\mbs^{R}\in\{0,1\}^{(m-m')\times 1}$ are the syndrome bits  corresponding to 
 		$\begin{bmatrix}
 			\mathbf{0} &\mbC
 		\end{bmatrix}$.
 		
 		\State  Verify the parity check condition related to the highly reliable bits.
 		\begin{algorithmic}
 			\If{$\mbs^{R}\neq  \mbC \big(\mbe^{R}\big)^\top$}
 			{} \Return ``Failure";
 			\Else
 			{}  we have a new parity-check condition
 			\begin{align}
 				\tilde{\mbs}=&  \Tilde{\mbH} \big(\hat{\mbe}'\big)^\top,  \label{eq:reduced_syndrome}
 			\end{align}
 			\State where      $\tilde{\mbs}=\mbs''-  \mbB\big(\mbe^{R}\big)^\top$. 
 			\If{(\ref{eq:reduced_syndrome}) is solvable}
 			{} \Return  $\tilde{\mbs}$,  $\Tilde{\mbH}$,  and $\sigma$;
 			\Else  
 			{} \Return ``Failure".
 			\EndIf
 			\EndIf

 		\end{algorithmic}
 		
 	\end{algorithmic}
 	
 \end{algorithm}

Equation~(\ref{eq:reduced_syndrome}) follows directly from the block decomposition
\[
\tau(\sigma(\mbH))=
\begin{bmatrix}
	\tilde{\mbH} & \mbB\\
	\mathbf 0 & \mbC
\end{bmatrix},
\]
which separates the contributions of the unknown variables $\mbe'$ and the fixed reliable variables $\mbe^R$.

{
	There are two possible failure modes in RSR:
	\begin{enumerate}
		\item[(i)] {Stage-1 failure:} the reliable bits are inconsistent with the syndrome, i.e.,
		$\mbs^{R}\neq \mbC (\mbe^{R})^\top$.
		
		\item[(ii)] {Stage-2 failure:} the reduced linear system
		$\tilde{\mbH}(\hat{\mbe}')^\top=\tilde{\mbs}$
		is rank-deficient or otherwise unsolvable.
	\end{enumerate}
Both conditions are explicitly verified in Step~3 of Algorithm~\ref{alg:RSR}.
The stage-1 failure corresponds to checking the sparse consistency condition
\[
\mbs^{R}= \mbC (\mbe^{R})^\top,
\]
which only requires a sparse matrix--vector multiplication and can be performed efficiently.
The stage-2 failure concerns the rank and pivot structure of $\tilde{\mbH}$.
In many applications, this information must be extracted in order to solve the reduced linear system, for example through Gaussian elimination or related procedures.
Therefore, this check can often be performed with little or no additional overhead when RSR is used together with a downstream decoder that requires such a step (e.g., OSD~\cite{PK21}, LSD~\cite{hillmann2025localized}).
}

 {
If RSR succeeds, the reduced system
\[
\tilde{\mbH} (\hat{\mbe}')^\top=\tilde{\mbs}
\]
is passed to a downstream decoder to obtain an estimate $\hat{\mbe}'$.
}
The full error estimate is recovered as
\[
\hat{\mbe}=\sigma^{-1}\big(\hat{\mbe}',\mbe^R\big).
\]

  \subsection{Reliability order based on BP$_4$} \label{sec:reliab}
  {To quantify the reliability of each error-variable estimate, we combine two sources of information:
  (i) the hard-decision trajectory of BP across all iterations, and
  (ii) the belief distribution produced at the final BP iteration.}

 Suppose that BP terminates without finding a valid solution after a maximum of $\Tau$ iterations.
 At this point, BP provides a belief distribution
 $\mbq_i=(q_i^I,q_i^X,q_i^Y,q_i^Z)$ for each qubit $i=1,\dots,n$,
 as well as a sequence of hard decisions
 $\mbh_i^{(j)}\in\{I,X,Y,Z\}$ for $j=1,\dots,\Tau$.
 Here $\mbh_i^{(j)}$ denotes the Pauli error assigned to qubit $i$ at iteration $j$.
 
 Our goal is to assign a reliability score separately for the $X$ and $Z$ components of each qubit.
 Following \cite{KKL23,PK21}, we introduce two complementary measures.

 \begin{definition} \label{def:LastFor} 
 	Define a   \emph{hard-decision  reliability vector} $\bm{\eta}=(\bm{\eta}_1,\dots,\bm{\eta}_n)\in\mathbb{Z}^n$ such that $\bm{\eta}_i$ 
 	represents the number of consecutive iterations during which the error at qubit $i$ remains unchanged until the final hard decision is made.
 \end{definition}

 \noindent 
 In other words, $\bm{\eta}_i$ is the length of the last run in the string $\mbh_{i}^{(1)},\mbh_{i}^{(2)},\dots, \mbh_{i}^{(\Tau)}$. 
 For example,  if the hard-decision outputs at the first qubit over the last five iterations are $X, Y, Z, Z, Z$, then the final hard decision output is $Z$ and  $\bm{\eta}_1=3$. This means that the error at qubit 1 persisted unchanged as $Z$ for three iterations when BP stops.
{A larger $\eta_i$ indicates that the estimate at qubit $i$
 has stabilized over a longer period and is therefore more reliable.}

 When BP fails, the belief distributions $\mbq_{i} = (q_i^I, q_i^X, q_i^Y, q_i^Z)$  will be utilized in the reliability measure. 
 \begin{definition}  
 	Define two   \textit{soft  reliability functions}
 	\begin{align*}
 		\phi^X(i) &= \max\{ q_{i}^X+q_{i}^Y, q_{i}^I+q_{i}^Z\},\\
 		\phi^Z(i) &= \max\{q_{i}^Z+q_{i}^Y, q_{i}^I+q_{i}^X\},
 	\end{align*}
 	for  $i\in\{1,2,\dots,n\}$.
 	
 \end{definition}
 \noindent  Here $\{q_i^X+q_i^Y,\;q_i^I+q_i^Z\}$ forms a binary distribution
 corresponding to whether an $X$ error is present on qubit $i$,
 and similarly $\{q_i^Z+q_i^Y,\;q_i^I+q_i^X\}$ corresponds to a $Z$ error.

 \begin{definition}
 	\label{def:highly_reliable}
 	An error bit $\mbe_i^a$, where $i\in\{1,\dots,n\}$ and $a\in\{X,Z\}$, is said to be \emph{highly reliable} if
 	its hard-decision reliability is either $\Tau$ or $\Tau+1$,
 	and its soft reliability $\phi^a(i)$ satisfies
 	\[
 	\phi^a(i)\ge \theta,
 	\]
 	where $\theta\in(0,1)$ is a predefined reliability threshold.
 \end{definition}

 An error variable has hard reliability $\Tau$ if its hard decision remains constant throughout all BP iterations.
 A value of $\Tau+1$ occurs when the variable stays equal to the identity operator $I$ from initialization through all iterations.
 The threshold $\theta$ is chosen close to one to ensure that a highly reliable coordinate is not only stable in its hard-decision history but also strongly supported by the final soft information.

 Now we define a reliability order for the bits in $\mbe	= \left[\begin{array}{c|c}
 	\mbe^X & \mbe^Z
 \end{array}\right] \in \{0, 1\}^{2n}$ based on the outputs of  BP$_4$, using the aforementioned hard and soft reliability functions.
   \begin{definition}
   	\label{def:reliability order}
   	For $1\le i,j\le n$ and $a,b\in\{X,Z\}$, the error bit $\mbe_i^a$ is said to be \emph{more reliable} than $\mbe_j^b$ if either
   	\[
   	\eta_i > \eta_j,
   	\]
   	or
   	\[
   	\eta_i=\eta_j \quad\text{and}\quad \phi^a(i)\ge \phi^b(j).
   	\]
   \end{definition}
   
    \noindent   That is, the stability of the hard-decision trajectory at the qubit level provides the primary ranking, while the final soft reliability acts as a tie-breaker between individual $X$ and $Z$ components.

 Our modified BP$_4$ algorithm 
is presented in Algorithm~\ref{alg:BP4}.

 { 
 	There exist various reliability metrics based on the (weighted) accumulation of log-likelihood ratio (LLR) signs across BP iterations \cite{GO05,GOK06,JZXZ07,rosseel2023sets}. Our hard-decision reliability can be viewed as a natural extension of this idea, focusing on the length of the most recent run in the hard-decision sequence.
 	The reliability order in Definition~\ref{def:reliability order} can also be adapted to incorporate alternative metrics, such as Entropy or Max \cite{PK21}. In these variants, the reliability of qubit $i$ is quantified either by the quaternary entropy of its belief distribution $\mathbf{q}_i$ (Entropy) or by the maximum probability within the distribution (Max).
 }

 \begin{algorithm}[ht] \caption{BP$_4$} 
 	\label{alg:BP4}
 	\textbf{Input}:  
 	check matrix $\mbH\in\{0,1\}^{m\times 2n}$,   syndrome $\mbs\in\{0,1\}^{m\times 1}$,       maximum number of iterations $\Tau$, {and initial error distributions $\mbp_{i}\in\mathbb{R}^4$ for $i=1,\dots,n$.}
 	
 	% % 
 	
 	\textbf{Output}: A valid error estimate $\hat{\mbe}\in\{0,1\}^{2n}$,
 	belief distributions  $\{\mbq_i\}_{i=1}^n$ and  hard-decision reliability vector $\bm{\eta}\in\mathbb{Z}^n$.
 	
 	\textbf{Initialization}:
 	\begin{algorithmic}
 		\State  
 		Let $\bm{\eta}_i=1$ for $i=1,\dots,n$.
 		Let $\mbh_i=I$ for $i=1,\dots,n$ be the initial hard decisions.
 		Let $\mbq_i=\mathbf{0}$ for $i=1,\dots, n$.
 	\end{algorithmic}
 	{\bf Steps}: 
 	\begin{algorithmic}
 		\For {$j = 1$ to $\Tau$}
 		\State  Update $\{\mbq_i\}_{i=1}^n$, using $\mbH, \mbs,$ and $ \{\mbp_i\}_{i=1}^n$.
 		\For {$i = 1$ to $n$}
 		\If{ $\mbh_{i}=\operatorname{HardDecision}(\mbq_i)$}
 		{} $\bm{\eta}_i\gets \bm{\eta}_i+1$;
 		\Else 
 		{}   $\mbh_i \gets  \operatorname{HardDecision}(\mbq_i)$. $\bm{\eta}_i=1$.  
 		\EndIf
 		\EndFor

 		\If{ the hard-decision ${\mbh}$ matches the syndrome $\mbs$}  \State \Return ``Success" and $\varphi(\mbh)$.
 		\EndIf
 		\EndFor
 		\State \Return  ``Failure", $\{\mbq_i\}_{i=1}^n$, and $\bm{\eta}$. \Comment{BP fails.}
 	\end{algorithmic}
 \end{algorithm}

\section{Order-statistic decoding   based on quaternary belief propagation } \label{sec:BPOSD}

{
	Once BP fails to produce a valid error estimate, combinatorial post-processing techniques are required. 
	These techniques operate on the residual errors left by BP and can include methods such as OSD, LSD, or other syndrome-based decoders. 
	In this section, we first review the basic OSD procedure, which relies on the reliability order defined in Section~\ref{sec:reliab}. 
	When combined with the RSR method introduced in Section~\ref{sec:rsr_decoding}, OSD can operate on a reduced system of variables, thereby lowering computational complexity.
}

In the context of OSD, a critical step involves identifying and sorting the accurate and reliable coordinates to establish $n+k$ linearly independent bits. Then, we can use these coordinates to generate a valid error that matches the syndrome. To accomplish this, we use the reliability orders defined in Section~\ref{sec:reliab}.

Let $\tilde{\mbe}$ be the hard-decision output from BP$_4$ after $\Tau$ iterations.
If $\tilde{\mbe}$  does not match the syndrome, OSD will be applied.
We now present our OSD algorithm based on the reliability orders
as in Algorithm~\ref{alg:OSD0},
which is   referred to as OSD$_4$-$0$.

\begin{algorithm}[ht] \caption{OSD$_4$-0} 
\label{alg:OSD0}

\textbf{Input}: 
 check matrix $\mbH\in\{0,1\}^{m\times 2n}$ of rank $n-k$,  syndrome $\mbs\in\{0,1\}^{m\times 1}$,  hard-decision reliability vector $\bm{\eta}\in\mathbb{Z}^n$,
       hard-decision vector $\tilde{\mbe}\in\{0,1\}^{2n}$, 
       and  belief distributions  $\{\mbq_i\}_{i=1}^n$.

    \textbf{Output}: A valid error estimate $\hat{\mbe}\in\{0,1\}^{2n}$.

    {\bf Steps}: 
    \begin{algorithmic}[1]\itemsep=6pt
    \State   Sort the error bits  in ascending order of reliability. 
 	Construct a  column permutation function $\pi$ corresponding to the sorting order. 
     Calculate $\pi(\tilde{\mbe} )$ and $\pi(  \mbH \Lambda )$.

    \State Perform Gaussian elimination on $\pi( \mbH \Lambda)$ and let the output be denoted by  $\rho(\pi( \mbH \Lambda))$, where $\rho$ represents the corresponding row operations in the Gaussian elimination. 
  If necessary, apply a column permutation function $\mu$ to ensure that the first $n-k$ columns of the resulting matrix  are linearly independent.  Thus we have 
	    \begin{align}
	        \mu (\rho(\pi (\mbH\Lambda )) ) = \left[
	    \begin{array}{cc}
	    \mbI_{n-k} & \mbA \\
	    \mathbf{0}       &  \mathbf{0}
	    \end{array}\right],  \label{eq:s'}
	    \end{align}
	for some $\mbA\in\{0,1\}^{(n-k)\times (n+k)}$, where the last $m-(n-k)$ rows   are all zeros.

 \State  Let $\mbs' =\rho(\mbs)$  so that the parity-check condition (\ref{eq:syndrome}) becomes 
 \begin{align}\mbs'=   \mu (\rho(\pi (\mbH\Lambda )) )    (\mu(\pi (\mbe)))^\top. \label{eq:modified_syndrome}
 \end{align} Remove the last $m-(n-k)$ entries of $\mbs'$. 

 \State  Let $ \mu(\pi(\tilde{\mbe}))= \left[\tilde{\mbe}^{U} ~~ \tilde{\mbe}^{R} \right]$,
   where $\tilde{\mbe}^{U}$~represents the independent and   unreliable subset of  $n-k$ bits and 
   $\tilde{\mbe}^{R}$ is the reliable subset of $n+k$ bits.  
   Use the reliable subset $\tilde{\mbe}^R$  and the modified 
   parity-check condition~\ref{eq:modified_syndrome} to select feasible unreliable bits~
   \begin{align}
   \tilde{\mbe}^U= {\mbs'}^\top \oplus  {\tilde{\mbe}^{R}} \mbA^\top . \label{eq:unreliable_subset}
   \end{align}
   \Return  \begin{align}
    \Hat{\mbe} \triangleq \pi^{-1}\left(\mu^{-1}\left(\left[{\mbs'}^\top \oplus  {\tilde{\mbe}^{R}} \mbA^\top \qquad \tilde{\mbe}^{R} \right] \right)\right). \label{eq:osd0}
\end{align} 
 
    \end{algorithmic}

\end{algorithm}

\subsection{OSD$_4$-$w$}\label{sec:osdw}
If some bits in the reliable subset are incorrect, additional processing is necessary. We may flip up to $w$ bits in the reliable subset and generate the corresponding error estimate. If this error estimate is valid, it is an error candidate. 
This process is repeated for all $\sum_{i=0}^w \binom{n+k}{i}$ possibilities, and the output will be a valid error candidate 
that is optimal with respect to $\delta$-ADD for a certain $\delta$.
This type of OSD algorithm is referred to as order-$w$ OSD$_4$ (or OSD$_4$-$w$ for short).

In order to explore all $\sum_{i=0}^w \binom{n+k}{i}$ possibilities,
we propose to use  the depth-first search (DFS) algorithm.
 	Figure~\ref{fig:dfs} illustrates the idea of  using DFS to traverse  all bit strings of length 4 and with  weight at most 2.
  One can see that a child node is obtained by flipping one zero bit from its parent node. Consequently, one can recursively generate all $\sum_{i=0}^w \binom{n+k}{i}$ possibilities by bit flipping and encoding.

\begin{lemma} \label{lemma: flip a bit}
Given an OSD$_4$-0 output vector, flipping one of its reliable bits
to generate an error candidate takes time $O(n)$.
Similarly, generating an error candidate corresponding to a child node from its parent node also takes  $O(n)$ time.
\end{lemma}
\begin{proof}

    Consider the vector $ \Big[\mbs'^\top\oplus\mbe^{R} \mbA^\top \, \Big| \tilde{\mbe}_{1}^{R}\, \tilde{\mbe}_{2}^{R}\, \cdots \, \tilde{\mbe}_{n+k}^{R}\Big]$
    generated by $\OSD_4$-0 in (\ref{eq:osd0}), where $\tilde{\mbe}_{j}^{R}$ are reliable bits. Suppose that the  matrix  $\mbA$ in  (\ref{eq:s'}) has columns $\mbA_{1}, \mbA_{2}, \dots, \mbA_{n+k}\in\{0,1\}^{n-k}$.
        Then according to~(\ref{eq:unreliable_subset}), flipping the bit $\tilde{\mbe}_{i}^{R}$ 
corresponds to adding the OSD$_4$-0 output by the vector 
\begin{align}\begin{bmatrix}
        	\mbA_j^\top & 0\cdots0 & 1&0\cdots 0
        \end{bmatrix}\in\{0,1\}^{2n}, \label{eq:flipping}
    \end{align}
        where $\mbA_j$ is the $j$-th column of $\mbA$,
        and generating the vector
    \begin{align}
            \left[\mbs'^\top \oplus \tilde{\mbe}^{R} \mbA^\top \oplus \mbA_i^\top \middle\vert \tilde{\mbe}_{1}^{R}\, \tilde{\mbe}_{2}^{R}\, \cdots \Big(1\oplus \tilde{\mbe}_{i}^{R}\Big) \cdots \, \tilde{\mbe}_{n+k}^{R} \right], 
    \end{align}
    which is a valid error candidate after permutations.
    This calculation  takes  $O(n)$ time.
\end{proof}

To efficiently navigate through these possibilities, we can utilize the technique described in Lemma \ref{lemma: flip a bit} combining with the DFS algorithm. Then all $\sum_{i=0}^w \binom{n+k}{i}$ candidates can be recursively generated.

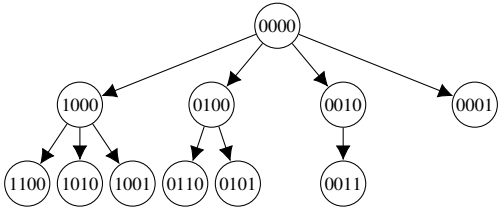
\begin{figure}
    \centering
\begin{tikzpicture}[scale=0.7,
 		level distance=1.5cm,
 		level 1/.style={sibling distance=2.5cm},
 		level 2/.style={sibling distance=1.0cm},
 		level 3/.style={sibling distance=0.5cm},
 		edge from parent/.style={draw, -{Latex[length=2mm, width=2mm]}, thin},
 		every node/.style={circle, draw, minimum size=6mm, inner sep=0}
 		]
 		\scriptsize
 		\node {0000}
 		child {node {1000}
 			    child {node {1100}}
 			child {node {1010}
 			}
     			child {node {1001}
 			}
 		}
 		child {node {0100}
 			child {node {0110}
 			}
                child {node {0101}
 			}
 		}	
 	child {node {0010}
                child {node {0011}}
 	}
        child {node {0001}
 	}
 	;
 	\end{tikzpicture}    \caption{ A tree with nodes consisting of all bit strings of length~4 and with at most weight 2. The root node is 0000. A child node is obtained by flipping one of the zero bits from its parent node that come after the last bit that is one.}
    \label{fig:dfs}
\end{figure}

\subsection{Complexity of OSD$_4$} \label{sec: complexity}

We analyze the complexity of  OSD based on $0$-ADD in the following.
The complexity of OSD$_4$-0 is dominated by the Gaussian elimination step, which is then $O(2n m^2) = O(n^3)$, assuming $m=O(n)$. For OSD$_4$-$w$ with $w > 0$, the total number of possible flips is $\sum_{j=0}^{w}\binom{n+k}{j}$, which can be expressed as $O((n+k)^w)$ or  simply $O(n^w)$.

In Lemma~\ref{lemma: flip a bit},   when we flip a bit in the reliable part, it requires only $O(n)$ calculations to find an error candidate. Additionally, by running through $\sum_{j=0}^{w}\binom{n+k}{j}$ possibilities using the DFS algorithm, the \mbox{OSD$_4$-$w$} algorithm has a complexity of $O(n^3+n^{w+1})$, which simplifies to $O(n^3)$ when $w \leq 2$.

  {
 \subsection{Reliable Subset Reduction combined with OSD}\label{sec:RSR_OSD}

 As introduced in Section~\ref{sec:rsr_decoding}, the RSR procedure removes highly reliable error variables   from the parity-check constraints, producing a reduced linear system. 
 This reduced system can then be solved using OSD or any other suitable combinatorial decoder to estimate the remaining unreliable variables. 
}

 Increasing the OSD order $w$ generally improves decoder performance. 
 In the extreme case, when $w=n+k$ and $\delta=n$, OSD$_4$-$(n+k)$ with $n$-ADD achieves maximum likelihood decoding. 
 However, larger $w$ comes at the cost of increased computational complexity.
 
 By reducing the problem size through RSR, we can lower the computational overhead in OSD, particularly at low error rates, as demonstrated in Section~\ref{sec:sim}. 
 Moreover, the reduced effective length allows the use of higher-order OSD under a fixed candidate budget. 
 For example, consider a candidate limit of $\Gamma$. 
 For the original problem, OSD$_4$-$2$ tests
 \begin{align}
 	\Gamma = \binom{n+k}{0} + \binom{n+k}{1} + \binom{n+k}{2}.
 	\label{eq:Gamma_osd_2}
 \end{align}
 If $v$ highly reliable bits are removed via RSR, the reduced system has $2n-v$ variables and rank $n-k-(m-m')$. 
 The maximum OSD order $w$ that respects the same candidate budget is then
 \begin{equation}
 	w = \max \left\{ x : \sum_{i = 0}^x \binom{n+k-v+(m-m')}{i} \leq \Gamma \right\}. \label{eq:w_eff}
 \end{equation}
 Even with the same number of candidates, the complexity is lower due to the reduced effective length.

\section{Approximate Degenerate OSD}\label{sec:ADOSD}

In this section, we show how code degeneracy can be exploited to prune the candidate set in OSD.  
Recall that in OSD$_4$-$w$, a list of
$
\sum_{i=0}^w \binom{n+k}{i}
$
candidates is generated by flipping subsets of reliable bits of the OSD$_4$-$0$ solution and adjusting the unreliable bits accordingly, and the best candidate is selected according to the $\delta$-ADD metric.

{
In quantum stabilizer codes, however, many of these candidates correspond to degenerate errors, i.e., they differ only by stabilizers and therefore have identical logical effects.  
Testing such degenerate candidates is unnecessary and can be avoided by tracking their logical syndromes.
}

{
Let $\mbL$ denote a binary logical matrix satisfying
$
\mbH \Lambda \mbL^\top = \mathbf{0},
$
so that the logical syndrome of an error $\mbe$ is given by
$
\mbL \Lambda \mbe^\top.
$

Let $\hat{\mbe}$ be the OSD$_4$-$0$ solution.  
As in (\ref{eq:s'}), after applying the same column permutation $\pi$, row operations $\rho$, and column permutation $\mu$ used in Gaussian elimination, the logical matrix becomes
\begin{align}
\mbL' = \mu\!\big(\rho(\pi(\mbL\Lambda))\big).
\end{align}

Flipping a single reliable bit   corresponds to adding the associated column of $\mbL'$ to the logical syndrome.  
Therefore, a reliable bit contributes to a nontrivial logical change if and only if the corresponding column of $\mbL'$ is nonzero.
}

 {
 \begin{theorem} [Degenerate candidate pruning] \label{thm:degeneracy_condition}
 	Consider an $[[n,k,d]]$ stabilizer code whose parity-check matrix has been transformed into the form
 	$
 	\mu(\rho(\pi(\mbH\Lambda))) =
 	\begin{bmatrix}
 		\mbI_{n-k} & \mbA
 	\end{bmatrix}
 	$
 	as in (\ref{eq:s'}), where the last $n+k$ columns correspond to the reliable bits in OSD$_4$-$0$.
 	Let $\mbL'$ be the transformed logical matrix defined above.
 	
 	\begin{enumerate}
 		\item In $n$-ADD OSD, only reliable bits whose corresponding columns in $\mbL'$ are nonzero need to be flipped when generating OSD$_4$-$w$ candidates.
 		
 		\item If all columns of $\mbL'$ corresponding to the reliable bits are zero, then all OSD$_4$-$w$ candidates are degenerate with the OSD$_4$-$0$ solution, and OSD$_4$-$w$ reduces to OSD$_4$-$0$ for any $w$.
 	\end{enumerate}
 \end{theorem}
}
 
 \begin{proof}
 	{
 	Flipping a reliable bit produces a new error candidate whose logical syndrome differs from that of $\hat{\mbe}$ by the corresponding column of $\mbL'$.  
 	In $n$-ADD, candidates that differ only by stabilizers are equivalent and need not be tested.
 }

 	1) Therefore, only reliable bits whose columns in $\mbL'$ are nonzero can generate candidates in different logical cosets, and only those flips must be considered.
 	
 	2) If all such columns are zero, then flipping any reliable bit leaves the logical syndrome unchanged, so all OSD$_4$-$w$ candidates are degenerate with $\hat{\mbe}$, and no higher-order OSD can improve upon OSD$_4$-$0$.
 \end{proof}

Exploiting  Theorem~\ref{thm:degeneracy_condition}-1) with  $n$-ADD
makes OSD$_4$-$w$ more effective. By incorporating with the RSR algorithm,
we can work on a   shortened check matrix $\tilde{\mbH}$ and higher order of OSD can be executed.
Consequently,  (\ref{eq:w_eff}) can be improved to
% \begin{equation*}
	$w = \argmax_x \sum_{i = 0}^x \binom{u}{i} \leq  \Gamma,$ 
	%\label{eq:w_eff2}
	% \end{equation*}
where $u\leq n+k-v+(m-m')$ is the number of columns in the shortened check matrix 
corresponding to nontrivial logical operators.

However, the complexity of $n$-ADD is not practically feasible. We may use $\delta$-ADD by choosing $\delta$ as the highest weight of a set of low-weight stabilizer generators. For the case of a rotated surface code~\cite{BM07}, the stabilizer generators are of weight $2$ or $4$, and the next lowest stabilizer weight is~$6$. Thus, $n$-ADD can be well approximated by $4$-ADD at error rates lower than $1\%$. Assume that OSD$_4$-0 generates a list of error candidates $\mathcal{E}$. Then, $4$-ADD outputs:
 \begin{align}
    \hat{\mbe}=\argmax_{\mbe} \sum_{i=0}^{n-k}\Pr \big\{ \mbe+\mbH_i \big\}, \label{eq:4-ADD}
 \end{align}
where $\mathbf{S}_0 = I$ and $\{\mathbf{S}_i\}_{i=1}^{n-k}$ is a set of stabilizer generators for the rotated surface code. 
For practical implementation, let $W(x)$ be the (Pauli) weight enumerator of the set $\{ \mathbf{E} + \mathbf{S}_i : i = 0, \dots, n-k \}$, and then the calculation of (\ref{eq:4-ADD}) can be approximated by using the dominating terms in $W(\epsilon/(1-\epsilon))$ for error rate $\epsilon < 1\%$.

    For codes with algebraic structures, the first few lowest-weight terms can be derived, but determining the complete weight enumerator for a general code is an NP-hard problem. As an alternative, we propose selecting a set of $m \geq n-k$ low-weight stabilizer generators and calculating the weight distribution of ${m \choose c}$ stabilizers, where $c$ is a constant independent of $n$. This method provides a good approximation of the dominant terms for practical purposes. To maintain an overall complexity of $O(n^3)$, the value of ${m \choose c}$ should be chosen accordingly.

To achieve low computational complexity, OSD$_4$-$0$ is often used in practice.  
Theorem~\ref{thm:degeneracy_condition} shows that when $n$-ADD is employed, only flips that induce nontrivial logical syndromes need to be considered.  
However, this argument does not directly apply to OSD$_4$-$0$.
In $0$-ADD, candidates are compared solely by their weight.  
Multiplying an error by a stabilizer may change its weight while leaving its logical coset unchanged, and therefore a degenerate error can become a better candidate under $0$-ADD.  
As a result, in OSD$_4$-$0$ one cannot restrict bit flips solely based on logical syndromes as in Theorem~\ref{thm:degeneracy_condition}-1).  
Nevertheless, the degeneracy condition in Theorem~\ref{thm:degeneracy_condition}-2) still applies.

We therefore introduce a weaker but computationally efficient sufficient condition for pruning degenerate candidates.

 \begin{corollary} \label{cor:OSD0}
 Let $\tilde{S}\in\{0,1\}^{m'\times (2n-v)}$ be a shortened check matrix by the RSR algorithm. Suppose that   $\tilde{S}$ can be transformed into 
 $\begin{bmatrix}
     \mbI& \tilde{\mbA}
 \end{bmatrix}$ after Gaussian elimination and a column permutation,
 where $\tilde{\mbA}$ is associated with the reliable bits.
If all the columns of $\tilde{\mbA}$  are of weight less than $d-1$, then OSD$_4$-$w$ on the corresponding reduced linear system is equivalent to   OSD$_4$-0 for any $w$.

 \end{corollary}

 	 {
 	 \subsection{RSR + ADOSD}

If many bits are identified as highly reliable, the number of stabilizers associated with these bits would be high. Consequently, the shortened check matrix $\tilde{\mathbf{S}}$ from the RSR algorithm will have fewer rows and lower column weights.  Consequently, the sufficient condition in Corollary~\ref{cor:OSD0} is satisfied with high probability at low error rates, making higher-order OSD unnecessary.
}

 We combine RSR, degeneracy-aware pruning, and order-adaptive OSD into a unified approximate degenerate OSD (ADOSD) framework, as detailed in Algorithm~\ref{alg:ADOSD}, using $0$-ADD. 
   {Note that if decoding performance is of higher priority
   or if the problem size is small enough  (e.g., after RSR reduction), $\delta$-ADD can be utilized and can be integrated into Algorithm~\ref{alg:ADOSD}.
}

\begin{algorithm}[ht] \caption{ADOSD$_4$} \label{alg:ADOSD}
    \textbf{Input}:
  check matrix $\mbH\in\{0,1\}^{m\times 2n}$ of rank $n-k$, syndrome $\mbs\in\{0,1\}^{m\times 1}$, 
        BP output distributions $\{\mbq_i\}_{i=1}^n$, 
hard-decision reliability vector $\bm{\eta}\in\mathbb{Z}^n$,   
  belief distributions   $\{\mbq_i\}_{i=1}^n$,  soft reliability threshold $\theta$,
      hard-decision vector $\mbe\in\{0,1\}^{2n}$,
      maximum number of OSD flips  $\Gamma$, and   backup OSD order\,$w$.

    \textbf{Output}: A  valid error estimate $\hat{\mbe}\in\{0,1\}^{2n}$.

    {\bf Steps}:

      \begin{algorithmic}
      \State     Apply the RSR algorithm.  
  \If{ RSR returns failure,}
  \State  apply OSD$_4$-$w$  to   the original problem, and  \Return $\hat{\mbe}$;
  \Else  {} RSR returns $\mbs'$, $\tilde{\mbH}$, $\mbe^{R}$, and $\sigma$. 
  	Suppose $\tilde{\mbH}$ in reduced row echelon form  is $[I|\tilde{A}]$. 
   \If{all the columns of $\tilde{A}$   have weight less than $d-1$,} 
   {} $w \gets 0$;   
      \Else  {}  determine $w$ by (\ref{eq:w_eff}).
        \EndIf
\State     Use OSD$_4$-$w$ on $\mbs'$, $\tilde{S}$, and the given data to find an error estimate $\hat{\mbe}'\in\{0,1\}^{2n-v}$, and  
 \Return $\sigma^{-1}\left(\begin{bmatrix}
         \hat{\mbe}'&\mbe^{R}
         \end{bmatrix}\right)$.

         \EndIf
   
    \end{algorithmic}
     
 \end{algorithm}

{
\section{Extension to Circuit-Level Noise via Detector Error Models} \label{sec:DEM}
Our framework is not confined to code-capacity noise; it admits a natural extension to realistic circuit-level noise through detector error models, as we describe below.

At the circuit level, gates, measurements, and ancilla qubits are faulty, and quantum error correction typically involves multiple rounds of stabilizer syndrome extraction circuits. 
The decoding task is to infer the most likely error configuration given the circuit description and the observed measurement outcomes. 
This circuit-level decoding problem can be formulated as solving a linear system of equations~\cite{Pry20,KL24b}. 
Specifically, the \emph{error syndrome} of a location error is the set of measurement outcomes produced when only that error occurs in the syndrome extraction circuits, while the \emph{logical syndrome} is given by the commutation relations between the logical operators and the residual data-qubit error after circuit evolution. 
Collecting the error and logical syndromes of all basis location errors yields a circuit-level parity-check matrix. 
For an $n$-qubit stabilizer code, the resulting circuit-level decoding problem typically has size $O(n^3)$
using $O(\sqrt{n})$ syndrome extraction rounds.

STIM~\cite{Gid21stim} provides a detector error model (DEM) for a given syndrome-extraction circuit, which yields a binary circuit-level parity-check matrix $\mbH_{\mathrm{DEM}}$ and a corresponding binary logical matrix $\mbL_{\mathrm{DEM}}$. The error syndrome is the binary vector of detector variables that represent the syndrome difference between two consecutive rounds.
Moreover, location errors that produce identical error and logical syndromes are indistinguishable and are therefore merged by STIM into equivalence classes, resulting in a nonuniform probability vector $\mbp_{\mathrm{DEM}}$ that specifies the approximate probability of each error class. 
Since STIM is a Clifford circuit simulator, it records only commuting observables; in particular, either logical $X$ operators or logical $Z$ operators are monitored. 
In the following, we assume that logical $X$ operators are observed.

Let $N$ denote the number of distinct error classes. 
This leads to a binary decoding problem: let $\hat{\mbe}\in\{0,1\}^N$ denote the decoder output for an actual error $\mbe\in\{0,1\}^N$. 
The decoder output is classified as a \emph{logical $X$ error} if
\begin{align}
	\mbH_\mathrm{DEM} \Lambda (\mbe+\hat{\mbe})^\top &= \mathbf{0}, \label{eq:cl_H}\\
	\mbL_\mathrm{DEM} \Lambda (\mbe+\hat{\mbe})^\top &\neq \mathbf{0} \label{eq:cl_L}.
\end{align}

This abstraction maps a fault-tolerant stabilizer circuit to a binary parity-check matrix with independent error probabilities, which allows our decoder to be applied without modification. 
For our purposes, this binary decoding problem can be embedded into the quaternary stabilizer formalism by treating $\mbH_{\mathrm{DEM}}$ as $Z$-type stabilizers and $\mbL_{\mathrm{DEM}}$ as logical $Z$ operators, while neglecting commutation relations.

Importantly, circuit-level DEMs are extremely sparse and highly biased: most columns correspond to single-fault events with very small support.
As a result, when BP is applied to $\mbH_\mathrm{DEM}$, a large fraction of error variables become highly reliable, making RSR particularly effective.
In the low-noise regime, RSR typically reduces the effective system size by one to two orders of magnitude, transforming a nominal $O(n^3)$ circuit-level problem into a reduced system whose size scales roughly linearly with the code distance when error rate is low.
}

\section{Simulation results} \label{sec:sim}

We simulate the performance of the MBP$_4$+ADOSD$_4$ scheme on LCS codes~\cite{ORM24}, BB codes~\cite{BCG+24}, and various 2D topological codes, including toric, rotated surface, color, and twisted XZZX codes~\cite{Kit03,BM06,BM07,HFDM12,THD12,KDP11,ATBFB21} (see \cite[Table~I]{KL22isit}) {under the code-capacity noise model, and we further simulate rotated surface codes under a circuit-level noise model.}

The decoding performance of a quantum code $\cC(\cS)$ and a decoder is assessed through the \textit{logical error rate} (LER).
For a code with distance $d$, errors up to $t=\lfloor\frac{d-1}{2}\rfloor$ are correctable, while some errors of weight $t+1$ become uncorrectable. Consequently, at low physical error rates, the performance curve of a decoder follows a scaling behavior of $O(\epsilon^{t+1})$. This implies that in a log-domain plot, the decoder’s performance curve should exhibit an error floor with a slope of $t+1$.
A good decoder will have a low error floor. However, if a decoder fails to correct all errors of weight up to $t$, the error floor will have a slope lower than $t+1$. In the worst case, a flatter error floor may appear earlier than expected.

Each data point in the plots is based on at least 100--1000 logical error events in the low error rate regime.
At LER $10^{-6}$, this represents an error bar less than $2\times 10^{-7}$, which is roughly the size of the markers for the data points. Thus the error bars are neglected for clarity.

We use only the parallel schedule for the following simulations of MBP$_4$+ADOSD$_4$.
In contrast, \cite{DGMSV23} showed that a random-order schedule can improve BP decoding performance.
Accordingly, we implement the AMBP decoder with a random-order serial schedule, achieving better performance than the results reported in \cite{KL22}.

In ADOSD$_4$ simulations, $\Gamma$ is chosen according to (\ref{eq:Gamma_osd_2}).

\subsection{Parameter Selection and Tuning}
We use the MBP$_4$ decoder introduced in \cite{KL22} as a predecoder for OSD. MBP$_4$ features a parameter $\alpha$ that controls the step size of message updates during BP iterations: when $\alpha < 1$, the step size is enlarged; when $\alpha > 1$, the step size is reduced.
For $\alpha = 1$, MBP$_4$ is equivalent to the refined BP$_4$ in \cite{KL20}.
We will specify the value of $\alpha$ only when $\alpha \neq 1$ is used.
The optimization of $\alpha$ is discussed in Section~\ref{sec:alpha_opt}.

  {The maximum number of BP iterations, $\Tau$, is set to $100$ and the soft reliability threshold to $\theta = 0.999995$ for simulations under the code-capacity noise model to ensure high accuracy, 
		while for the circuit-level noise model, $\Tau$ is set to $10$ and $\theta = 0.99$ to prioritize computational efficiency.}

\subsubsection{Optimization of the parameter $\alpha$  in MBP$_4$ for ADOSD} \label{sec:alpha_opt}

In this subsection, we demonstrate how to optimize the parameter $\alpha$ of our MBP$_4$+ADOSD$_4$ scheme using LCS codes~\cite{ORM24}. 
For most nondegenerate codes, $\alpha = 1$ is sufficient to achieve good decoding performance. However, for degenerate codes, selecting $\alpha > 1$ can result in more stable message passing, leading to improved OSD performance. This optimization should be conducted for each   code and potentially for each   error rate, though a single value of $\alpha$ often applies across an entire code family.

LCS codes are a specialized type of quantum codes that combines the features of surface codes and lift-product constructions. Decoding LCS codes presents   the most challenging problem for AMBP$_4$ among all the quantum codes considered in this paper. We find that the AMBP$_4$ decoder~\cite{KL22} is ineffective for LCS codes, despite its superior performance over binary BP-OSD decoders~\cite{PK21,RWBC20} on all other codes. 

\begin{figure}[htbp]
	\centering
	\includegraphics[width=0.49\textwidth]{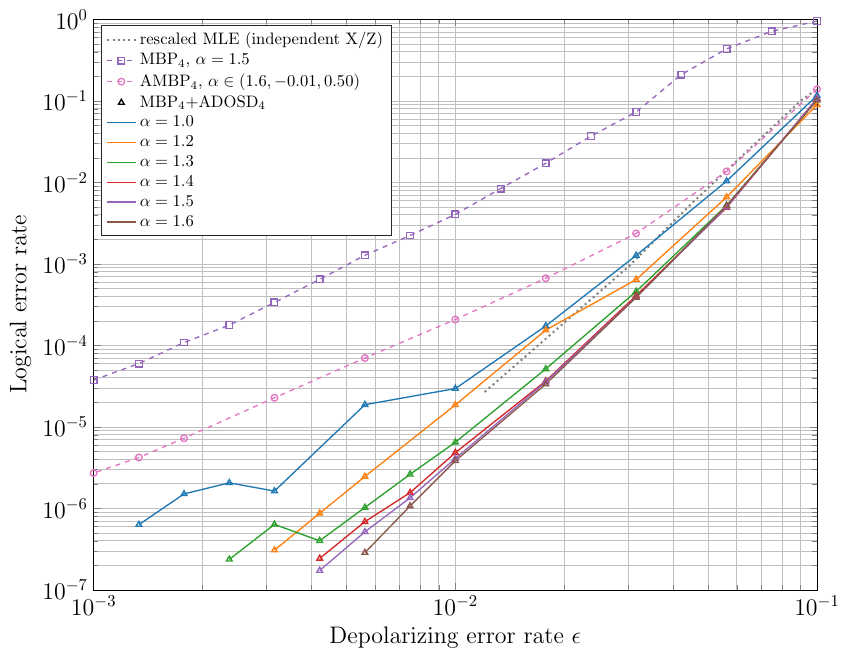}
	\caption{MBP$_4$+ADOSD$_4$ decoding with various values of $\alpha$ for the $[[175,7,7]]$ LCS code. The curve MLE is taken and rescaled from~\cite{ORM24}.
		The notation $\alpha\in(1.6,-0.01,0.5)$ means that $\alpha$ is tested in the sequence $1.60,1.59,1.58,\dots,0.51,0.50$ {with  $\Tau=100$ and $\theta=0.999995$.}}
	\label{fig:LCS175_alpha}
\end{figure}
Figure~\ref{fig:LCS175_alpha} compares various decoders on the $[[175,7,7]]$ LCS code, including an MLE (independent X/Z) reference curve from~\cite{ORM24}.
{
	To compare with the MLE curve under $X$-only noise, we rescale the physical error rate by a factor of $3/2$ so that the effective $X$ error probability matches that under depolarizing noise.
}
This MLE curve serves as a benchmark, as a decoder capable of correcting errors of weight up to $3$ should exhibit a performance curve parallel to it at low error rates.

The performance of MBP$_4$+ADOSD$_4$ on the $[[175,7,7]]$ LCS code, with varying values of $\alpha$, exhibits fluctuations at lower values of $\alpha$ but stabilizes as $\alpha$ increases  to 1.6. With $\alpha=1.6$, its performance curve indicates that it can correct nearly all weight-3 errors at LER $10^{-7}$.

 In~\cite{ORM24}, the BP-OSD-CS60 decoder~\cite{RWBC20} was shown to closely match the MLE curve on this code. Since the decoding of $X$ and $Z$ errors is performed independently, the MLE curve represents optimal performance under the assumption that $X$ and $Z$ errors are uncorrelated.
As MBP$_4$+ADOSD$_4$ leverages the $X/Z$ correlations, it significantly outperforms both the MLE curve and BP-OSD-CS60 by more than 15 times.
 In contrast, AMBP$_4$   encounters an error floor   for $\epsilon \leq 0.03$.

We observe a similar performance of MBP$_4$+ADOSD$_4$ for $[[15,3,3]]$, $[[65,5,5]]$ and $[[369,9.9]]$ LCS codes as well, while AMBP has an even higher error floor around LER $10^{-2}$ for the $[[369.9,9]]$ LCS code and a lower error floor around $10^{-5}$ for the $[[65,5,5]]$ LCS code.

\subsubsection{Robustness of the maximum iterations in MBP$_4$+ADOSD$_4$} \label{sec:iteration}

\begin{figure}[htbp]
    \centering
    \includegraphics[width=0.49\textwidth]{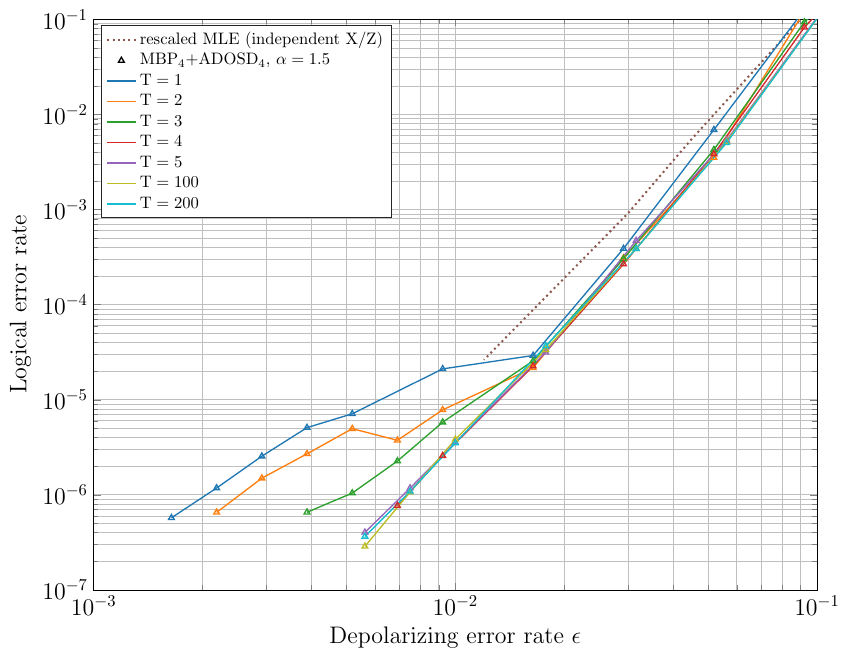}
    \caption{MBP$_4$+ADOSD$_4$ decoding with various maximum number of iterations $\Tau$ for the $[[175,7,7]]$ LCS code {with $\theta=0.999995$}.}
    \label{fig:LCS175_iter}
\end{figure}

Next, we demonstrate that the performance of ADOSD$_4$ remains robust even with variations in the maximum number of iterations in MBP$_4$. We note that in ADOSD$_4$, OSD-$w$ is invoked  when RSR fails. As a result, the overall algorithm maintains high accuracy.

Figure~\ref{fig:LCS175_iter} illustrates the performance of MBP$_4$+ADOSD$_4$ decoding for the [[175,7,7]] LCS code with varying maximum numbers of iterations, $\Tau =1,2,3, 4, 5,   100, 200$. The curve for $\Tau=4$ is already good enough and smooth.
While increasing the number of iterations results in slightly improved performance, the curves converge closely around an LER of $2 \times 10^{-7}$ to $3 \times 10^{-7}$.

\begin{table}[htbp]
	\begin{tabular}{|c|c|c|c|c|c|c|c|c|}
		\hline
		$\epsilon$& 0.1&0.056&0.032&0.018&0.01&0.007&0.0056\\
		\hline
		Avg. iter.& 194.46&97.67&23.51&6.31&2.44&1.70&1.23\\
		\hline
	\end{tabular}
	\caption{The average number of iterations for MBP$_4$  decoding on the $[[175,7,7]]$ LCS code with $\Tau=200$.}\label{tb:avg_iterLCS}
\end{table}

Table~\ref{tb:avg_iterLCS} presents the average number of iterations for $\Tau=200$. It can be observed that the average number of iterations rapidly decreases to one or two when $\epsilon<0.01$. Since we are targeting lower error rates, to balance accuracy and efficiency, we typically choose $\Tau=100$.

 \subsection{Code capacity noise model decoding simulations}
\subsubsection{Bivariate bicycle codes} \label{sec:BB_sim}

 BB codes are a class of QLDPC codes  proposed in~\cite{BCG+24}, featuring a high error threshold and a significantly higher code rate than topological codes. Several attempts at decoding this code family have been presented in~\cite{iOM24closed, GCR24} under the code capacity noise model. However, the BP+CB decoder in~\cite{iOM24closed} does not perform as well as BP-OSD-0. The BP+GDG decoder in~\cite{GCR24} shows improved performance, surpassing BP-OSD-0 and being comparable to BP2-OSD-CS10  over independent X errors,  as shown in \cite[Figure 4]{GCR24}.

Figure~\ref{fig:AMBP_BB} presents the simulation results of MBP$_4$+ADOSD$_4$ and AMBP$_4$ decoding for BB codes, both of which demonstrate significantly better performance than BP+GDG decoding~\cite{GCR24}. 
For instance, consider the [[144,12,12]] BB code. At an $X$ error rate of $0.02$ (corresponding to a depolarizing rate of about $0.03$), BP+GDG achieves an LER of approximately $10^{-4}$. In contrast, both MBP$_4$+ADOSD$_4$ and AMBP$_4$ achieve an LER below $4 \times 10^{-6}$ at the same depolarizing rate.
Moreover, MBP$_4$+ADOSD$_4$ outperforms AMBP$_4$.

For reference, we include dotted $O(\epsilon^{t+1})$ curves for each code. The $[[72,12,6]]$ and $[[90,8,10]]$ codes closely follow these reference curves, while the other three codes exhibit lower error floors below LER $10^{-7}$.

\begin{figure}
	\centering
	\includegraphics[width=0.49\textwidth]{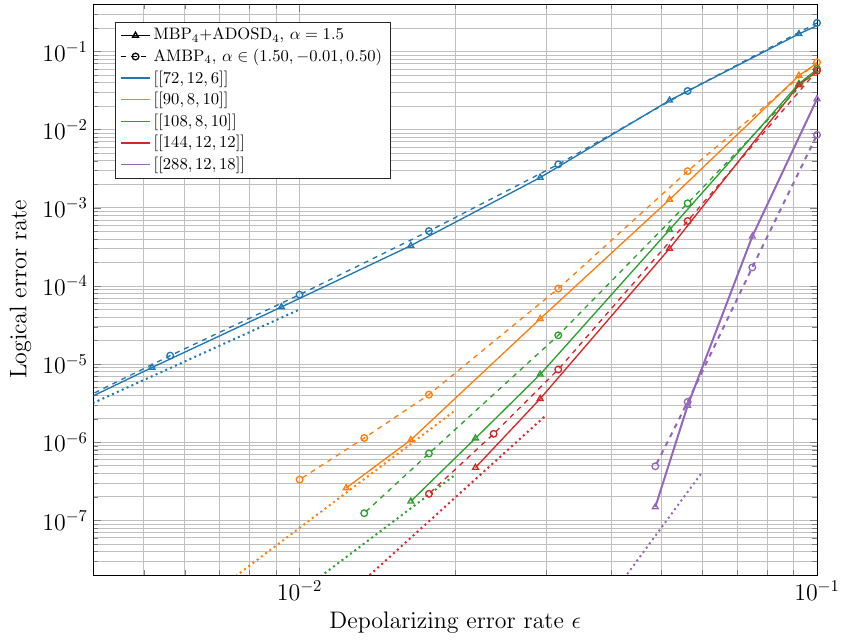}
	\caption{ AMBP$_4$ and MBP$_4$+ADOSD$_4$ decoding performance for various BB codes 
		{with  $\Tau=100$ and $\theta=0.999995$}.
%	\cnote{update later}
}
	\label{fig:AMBP_BB}
\end{figure}

\subsubsection{2D topological codes}\label{sec:2D_sim}

\begin{figure}
	\centering
	 \includegraphics[width=0.49\textwidth]{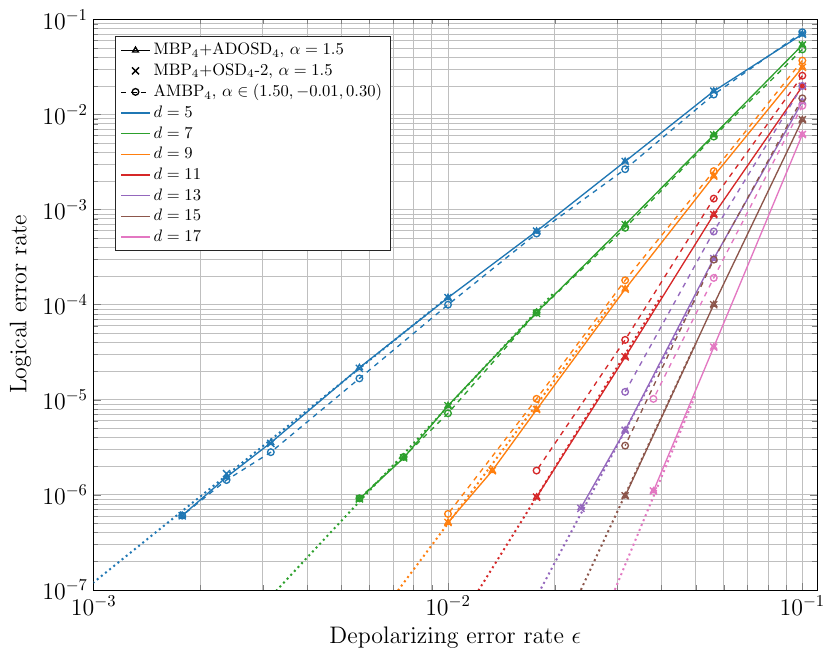}
	\caption{
	  MBP$_4$+OSD$_4$-2, MBP$_4$+ADOSD$_4$, and AMBP$_4$ decoding performance for various  $[[d^2,1,d]]$ rotated surface codes {with  $\Tau=100$ and $\theta=0.999995$}.
	}\label{fig:rot_surface}
\end{figure}

\begin{figure}
	\begin{subfigure}{0.49\columnwidth}
		\centering
		\centering\includegraphics[width=0.99\columnwidth]{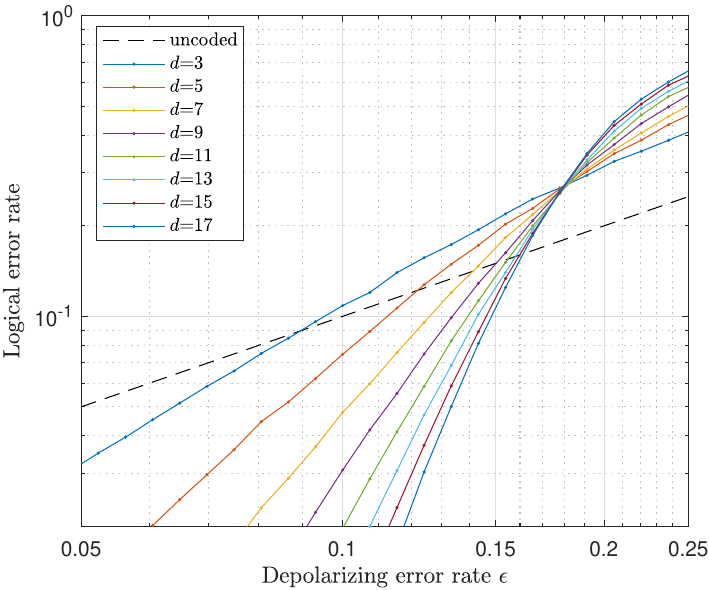}
		\caption{}
		\label{fig:3.6(a)}
	\end{subfigure}
	\hfill
	\begin{subfigure}{0.49\columnwidth}
		\centering
		\centering\includegraphics[width=0.99\columnwidth]{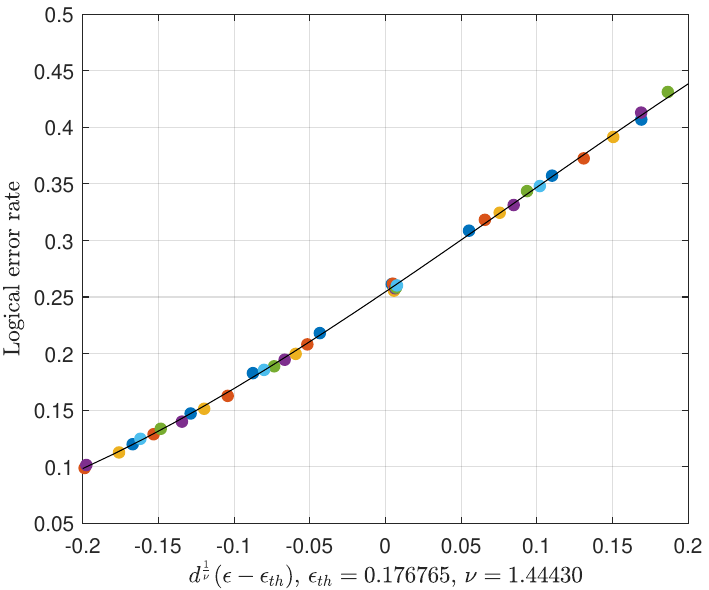}
		\caption{ }
		\label{fig:3.6(b)}
	\end{subfigure}

	\caption{
		(a) The threshold of MBP$_4$+OSD$_4$-2 on the $[[d^2,1,d]]$ rotated surface codes  {with $\alpha=1.0$, $T=100$, and $\theta=0.999995$} is about $17.67\%$ from the simulations. For each data point, we collected at least 10,000 logical error events. The dashed line stands for no error correction. 
		(b)  The accuracy of the threshold estimation about \BPOSD-2 on the rotated surface codes can be demonstrated using the critical scaling approach. The estimation is based on data from $d = 7$ to $d = 17$, ensuring that the results are not influenced by biased information.
	}\label{fig:surface_osd_thres}
	
\end{figure}

Here, we evaluate the decoding performance of our BP+OSD schemes on various 2D topological codes. Examining LER around $10^{-6}$ is important, as it helps identify whether an iterative decoder exhibits error floor behavior.

Figure~\ref{fig:rot_surface} presents the performance curves of MBP$_4$+ADOSD$_4$ and MBP$_4$+OSD$_4$-2 on $[[d^2,1,d]]$ rotated surface codes with distances ranging from 5 to 17. They exhibit  excellent performance, closely following the dotted reference curves.

The threshold analysis for $[[d^2,1,d]]$ rotated surface codes, conducted using the scaling ansatz method~\cite{WHP03,C21}, is shown in Figure~\ref{fig:surface_osd_thres}. The results show a threshold of approximately 17.67\% for BP$_4$+OSD$_4$-2, with MBP$_4$+ADOSD$_4$ exhibiting nearly identical performance. Notably, they outperform AMBP$_4$, particularly for larger $d$, resulting in a higher threshold compared to the ~16\% threshold of AMBP$_4$.

In summary, our decoding schemes demonstrate superior performance on surface codes compared to the MWPM decoder~\cite{Edm65,WFSH10}, both in terms of LER performance and error threshold (see, e.g.,~\cite{demarti2024decoding}).

Similar analyses are performed for toric, twisted XZZX, (6.6.6), and (4.8.8) color codes using both BP$_4$+OSD$_4$-2 and BP$_4$+ADOSD$_4$. The summarized threshold results are presented in Table~\ref{tb:AvsO_thres}. 
 ADOSD$_4$ and OSD$_4$-2  exhibit nearly identical performances on toric, surface, and twisted XZZX codes, with OSD$_4$-2 performing slightly better on the (6.6.6)  and (4.8.8) color codes.

 \begin{table}%[h]%[h!b] 
 	\caption{Thresholds for 2D topological codes under depolarizing errors.
 		The parameters {are $\alpha=1.0$,  $T=100$, and $\theta=0.999995$} } \label{tb:AvsO_thres} 
 	\centering{
 		$\begin{array}{|l|l|l|l|}
 			\hline
 			\text{code family} 	& \text{MBP$_4$+OSD$_4$-2}     &\text{MBP$_4$+ADOSD$_4$}     & \text{AMBP$_4$}\\
 			\hline   
 			\text{rotated toric}        & 17.52\%     &17.52\%   & \approx 17.5\% \\
 			\text{rotated surface}      & 17.67\%    &17.67\%    & \approx 16\% \\
 			\text{(6.6.6) color}   & 15.41\%     &15.18\%   & \approx 14.5\% \\
 			\text{(4.8.8) color} & 15.09\%     &14.69\%  & \approx 14.5\% \\
 			\text{XZZX twisted} & 17.72\%     &17.72\%   & \approx 17.5\% \\
 			\hline
 		\end{array}$
 	}
 	\\[3pt]
 	\begin{flushleft}
 	\end{flushleft}
 \end{table}

    {
    	\subsection{Circuit-Level Surface Code Decoding}

    	We simulate circuit-level decoding of the $[[d^2,1,d]]$ rotated surface codes using DEM matrices generated by STIM~\cite{Gid21stim}. For a distance-$d$ code, we simulate $d$ rounds of syndrome extraction.
    	
    	All faulty locations are assumed to have the same error rate $\epsilon$. We use 
    	\texttt{stim.Circuit} with experiment     	 \texttt{surface\_code:rotated\_memory\_z} to generate the syndrome-extraction circuits, and convert the resulting DEM into the binary matrices $\mbH_{\mathrm{DEM}}$, $\mbL_{\mathrm{DEM}}$, and the probability vector $\mbp_{\mathrm{DEM}}(\epsilon)$, where $\mbp_{\mathrm{DEM}}(\epsilon)$ depends on the physical error rate~$\epsilon$.
    	
    	In the default setting, the first round of detector variables contains detectors for only the $\frac{d^2-1}{2}$ $Z$ stabilizers. The subsequent $d-1$ rounds contain detector variables for all $d^2-1$ $X$ and $Z$ stabilizers. In addition, there is a final round of detector variables for the $Z$ stabilizers, which can be regarded as noiseless syndrome measurements. As a result, the default STIM circuit for a $[[d^2,1,d]]$ rotated surface code yields a total of $d(d^2-1)$ parity checks.
    	
    	Note that decoders which treat $X$ and $Z$ errors separately, such as the standard MWPM decoder, do not make use of both types of detector variables. For comparison, we therefore also generate decoding problems in which the $X$-type detector variables are removed. The parameters of the resulting DEM decoding problems are summarized in Table~\ref{tb:DEM_parameters}.

Since the default DEM contains more detector variables and hence more syndrome bits, it produces a larger number of error classes, i.e., a larger $N$, because some $Y$ errors can be distinguished from $Z$ errors by the presence of $X$-type detector variables.

    }

{
\begin{table*}[htbp]
	\centering
	\footnotesize
	\begin{subtable}[t]{0.45\textwidth}
		\centering
		\begin{tabular}{|c|c|c|c|c|}
			\hline
			Distance & 3 & 5 & 7 & 9 \\
			\hline
			$N$ & 219 & 1677 & 5471 & 12705 \\
			\hline
			\# Checks & 24 & 120 & 336 & 720 \\
			\hline
		\end{tabular}
		\caption{DEM decoding parameters with both $X$ and $Z$ detector  variables.}
	\end{subtable}
	\hfill
	\begin{subtable}[t]{0.45\textwidth}
		\centering
		\begin{tabular}{|c|c|c|c|c|}
			\hline
			Distance & 3 & 5 & 7 & 9 \\
			\hline
			$N$ & 55 & 301 & 883 & 1945 \\
			\hline
			\# Checks & 16 & 72 & 192 & 400 \\
			\hline
		\end{tabular}
		\caption{DEM decoding parameters with only $Z$ detector variables.}
	\end{subtable}
	
	\caption{Parameters of the DEM decoding problems for rotated surface codes. Subtable (a) shows the default setting with both $X$ and $Z$ detector variables, and subtable (b) shows the modified setting with only $Z$ detector variables.}
	\label{tb:DEM_parameters}
\end{table*}
}

{
In our simulations, an error vector $\mbe$ is sampled according to $\mbp_{\mathrm{DEM}}(\epsilon)$
and its syndrome $\mbH_{\mathrm{DEM}} \mbe^\top$ is computed.
Given $\mbH_{\mathrm{DEM}}$, $\mbL_{\mathrm{DEM}}$, the measured syndrome $\mbH_{\mathrm{DEM}} \mbe^\top$, and $\mbp_{\mathrm{DEM}}(\epsilon)$, we apply MBP$_4$+ADOSD$_4$ to this binary decoding problem to obtain an estimate $\hat{\mbe}$.
The logical $X$ error rate is then estimated using (\ref{eq:cl_H}) and (\ref{eq:cl_L}).

We collected 1000 logical error events for each data point in the plots, except for the cases with
 $d=9$ and $N=12705$ at $\epsilon=0.002$ and $\epsilon=0.001$. 
}

%\subsubsection{Decoding performance}
{
For computational efficiency, the maximum number of BP iterations is set to $T=10$.
In particular, we use only soft reliability with threshold $\theta = 0.99$, without employing hard-decision history.
By lowering the reliability threshold, RSR produces a smaller reduced decoding problem.

The step size is fixed to $\alpha=1.5$, following the choice used in the code-capacity noise model without further optimization.

\begin{figure}
	\centering
	\includegraphics[width=0.49\textwidth]{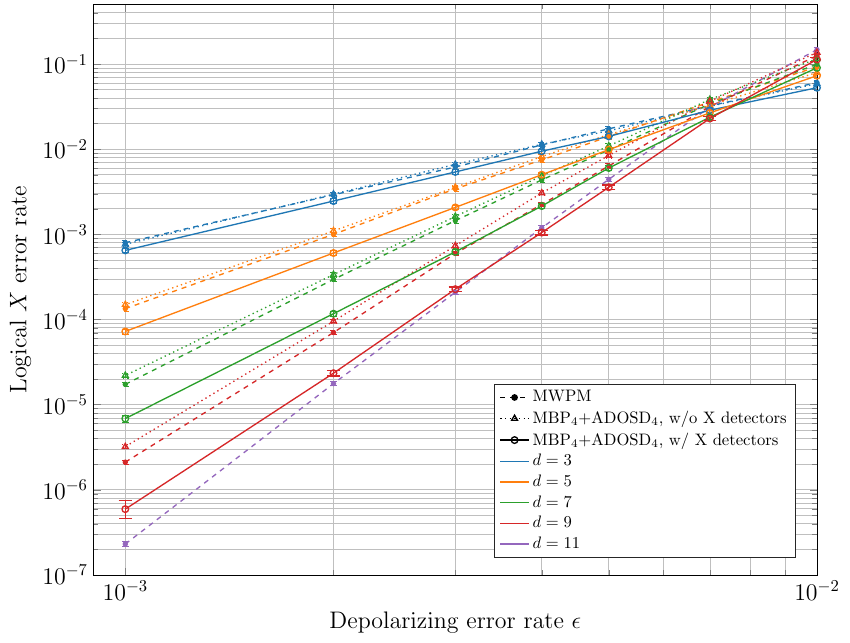}
	\caption{
		MBP$_4$+ADOSD$_4$ and MWPM decoding performance for various $[[d^2,1,d]]$ rotated surface codes
		with $\alpha=1.5$, $T=10$, and $\theta=0.99$ under the circuit-level noise model.
		Error bars represent 95\% confidence intervals.
	}
	\label{fig:rot_surface_cl}
\end{figure}

The simulation results are shown in Fig.~\ref{fig:rot_surface_cl}, where the MWPM decoder is implemented using PyMatching2~\cite{HG25} for comparison.
We observe that MWPM performs better than MBP$_4$+ADOSD$_4$ for DEM instances without $X$-type detector variables.
However, when $X$ detector variables are included, MBP$_4$+ADOSD$_4$ outperforms MWPM by exploiting the additional syndrome information.
Moreover, for higher physical error rates ($p > 0.003$), the $d=9$ MBP$_4$+ADOSD$_4$ curve outperforms the $d=11$ MWPM curve.
The error threshold is also improved to approximately $0.76\%$, compared with the $0.7\%$ reported for PyMatching2.

For reference, as shown in \cite{muller2025improved}, RelayBP performs worse than MWPM for rotated surface codes.
We also remark that MWPM can be adapted to exploit $X/Z$ correlations~\cite{Fow13,PF23,tian2025enhancing}.

Furthermore, we include a direct comparison with BP+LSD-0 from~\cite{hillmann2025localized} in Fig.~\ref{fig:rot_surface_d59}, where our decoder consistently achieves a lower logical error rate.
Although it is unclear whether~\cite{hillmann2025localized} uses the default DEM generated by stim, both of our performance curves (with and without $X$ detector variables) outperform BP+LSD-0.

\begin{figure}
	\centering
	\includegraphics[width=0.49\textwidth]{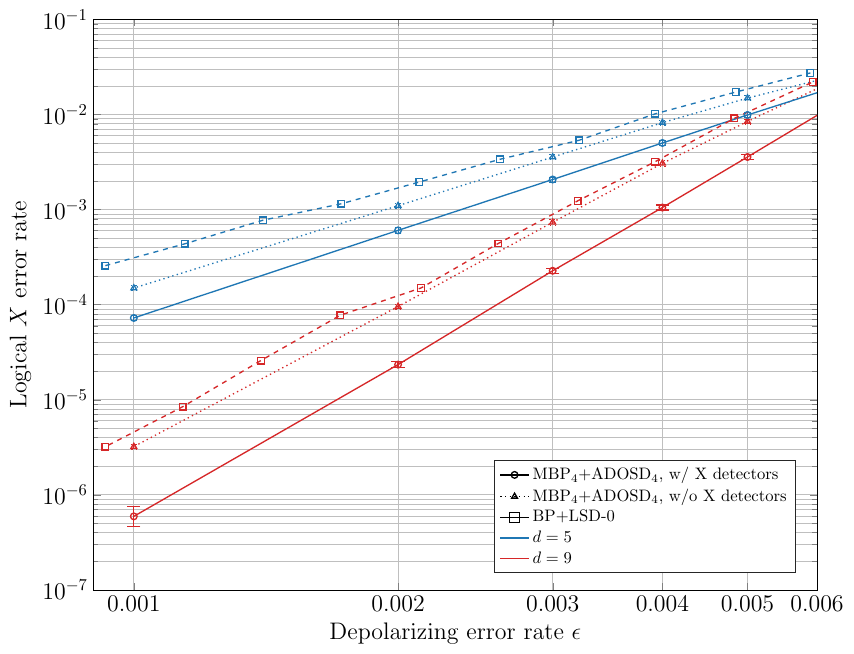}
	\caption{
		MBP$_4$+ADOSD$_4$ and BP--LSD-0 decoding performance for $[[d^2,1,d]]$ rotated surface codes
		with $\alpha=1.5$, $T=10$, and $\theta=0.99$ under the circuit-level noise model.
		The BP+LSD-0 curves are from~\cite{hillmann2025localized} and the algorithm uses a maximum of 30 BP iterations.
	}
	\label{fig:rot_surface_d59}
\end{figure}
}

\subsection{Decoding statistics}

We summarize the decoding statistics for circuit-level surface-code decoding of distance $9$ without $X$-type detector variables in Table~\ref{tb:statistics_d9}, using $\alpha=1.5$, $T=10$, and $\theta=0.99$.
{The row “Shots” reports the total number of Monte Carlo shots, while “OSD calls/Shots” is the percentage of shots for which BP does not converge to a valid solution and OSD is therefore invoked.
}
The row “Average length after RSR” shows that the effective code length after reduction is below $4\%$ of the original for $\epsilon \le 0.005$, and below $1\%$ for $\epsilon = 0.001$, demonstrating the strong reduction capability of RSR.

From the statistics, we observe no stage-1 failures and only a small percentage of stage-2 failures.
Even if all stage-2 failures are conservatively counted as logical errors, the overall logical error rate is only marginally affected.

The row “Avg iterations w/o OSD” corresponds to samples that are successfully decoded by BP alone without invoking OSD.
{Typically, only about $3$–$6$ BP iterations are required.}
This justifies setting the maximum number of BP iterations to $T=10$, which is sufficient for RSR to accurately identify reliable and unreliable bits.

For reference, in the reported simulations of Relay-BP-5, $R=601$ and $T_r=60$, resulting in more than $3.6\times10^4$ BP iterations~\cite{muller2025improved}.

Finally, we remark that the current software runtimes do not yet fully reflect the intrinsic speedup enabled by RSR. Our current research-oriented implementation embeds the binary DEM decoding problem into a quaternary formalism, which incurs unnecessary computational overhead. A dedicated binary implementation is expected to further reduce the wall-clock time.

%
%% [Thread 0] p_err=1.000000e-02  LER=1.396760e-01  +err=7.474650e-04  -err=7.474650e-04
%% [Thread 0] p_err=7.000000e-03  LER=3.810400e-02  +err=3.904049e-04  -err=3.904049e-04
%% [Thread 0] p_err=5.000000e-03  LER=8.650000e-03  +err=1.860108e-04  -err=1.860108e-04
%% [Thread 0] p_err=3.000000e-03  LER=7.450000e-04  +err=5.458938e-05  -err=5.458938e-05
%% [Thread 0] p_err=2.000000e-03  LER=9.600000e-05  +err=1.959592e-05  -err=1.959592e-05
%% [Thread 0] p_err=1.000000e-03  LER=5.000000e-06  +err=4.472136e-06  -err=4.472136e-06

 \begin{table}[ht]
 	\caption{Statistics of circuit-level surface decoding ($d=9, N=1945$) without $X$ detector variables. Parameters: $\alpha=1.5, T=10, \theta=0.99$. $n_{red}$ denotes the reduced variable length by RSR.}
 	\label{tb:statistics_d9}
 	\centering
 	\begin{tabular}{|l|c|c|c|c|c|c|}
 		\hline
 		Depolarizing rate $\epsilon$  &  0.01  & 0.007  &  0.005 & 0.003 & 0.002 & 0.001  \\
 		\hline
 		Shots & $10^6$ & $10^6$ & $10^6$ & $10^6$ & $10^7$ & $10^7$\\
 		OSD calls / Shots & $99.97\%$ & $96.85\%$ & $79.12\%$ & $36.88\%$ & $16.41\%$ & $3.90\%$ \\
 		\hline
 		Average length after RSR ($n_{red}$) & 229.53 & 113.27 & 62.34 & 32.85 & 22.67 & 14.99 \\
 		Relative code length ($n_{red}/N$) & 0.118 & 0.058 & 0.032 & 0.017 & 0.012 & 0.008 \\
 		\hline
 		Stage-1 Failures & 0 & 0 & 0 & 0 & 0 & 0 \\
 		Stage-2 Failures / Shots & $1.19\times 10^{-4}$ & $4.33\times 10^{-4}$ & $8.61\times 10^{-4}$ & $8.88\times 10^{-4}$ & $3.95\times 10^{-4}$ & $9.16\times 10^{-5}$ \\
 		\hline
 		Avg. BP iterations (converged) & 6.21 & 5.67 & 5.15 & 4.53 & 4.07 & 3.50 \\
 		\hline
 	\end{tabular}
 \end{table}

{
 The order-of-magnitude reduction in the residual problem size  ensures that the computational cost of the ADOSD stage remains comparable to only a few BP iterations. This low-complexity feature is consistent across both code-capacity and circuit-level noise models, facilitating the implementation of high-order OSD in large-scale quantum error correction systems.
}

{
\section{Related Work} \label{sec:related}
Our decoder essentially follows a BP+RSR+OSD pipeline: 
1) run BP; 
2) if BP does not converge to a valid codeword, use the reliability statistics generated by BP to perform RSR and obtain a reduced linear system; 
3) apply OSD to solve the reduced system. 
In this way, the postprocessing stage is significantly accelerated.

Related work such as Ambiguity Clustering (AC) \cite{WB24ambiguity} and Localized Statistics Decoding (LSD) \cite{hillmann2025localized} also aim to speed up the postprocessing stage. 
AC employs a divide-and-conquer strategy by decomposing the global decoding problem into multiple localized clusters. 
Each cluster is solved independently to determine its local logical effect, and the local solutions are then linearly combined to reconstruct a global logical correction.

LSD follows a cluster-growth paradigm inspired by the Union-Find (UF) decoder~\cite{DN17} to localize syndromes. 
It further uses linear-algebraic verification to govern cluster merging and syndrome satisfaction, allowing LSD to be applied to codes without a geometric structure.

Since AC and LSD can be viewed as high-efficiency alternatives to OSD for postprocessing, 
it is natural to integrate them into a unified pipeline such as BP+RSR+AC or BP+RSR+LSD. 
Note that RSR requires performing rank-revealing linear-algebraic operations on the reduced system.
Importantly, essentially the same class of linear-algebraic operations (e.g., pivoting, rank tests, and column selections) is also required by OSD, AC, and LSD in order to initialize clusters, perform algebraic verification, or generate candidate solutions.
Therefore, applying AC or LSD after RSR does not introduce a fundamentally new computational stage.

We further note that AC applies degenerate decoding when solving the local subproblems. 
Because these subproblems are much smaller, it becomes feasible to perform fully degenerate decoding by optimizing over logical cosets. 
This observation motivates the use of $\delta$-ADD with $\delta>0$ as a hard-decision rule in OSD after RSR has reduced the problem size sufficiently.
 
To summarize, degeneracy is exploited in our algorithm in two distinct ways.
(1) Since OSD-$w$ is a list decoder, we first use a degeneracy condition (Theorem~\ref{thm:degeneracy_condition} and Corollary~\ref{cor:OSD0}) to determine whether multiple logical cosets are possible, and hence whether generating a candidate list is necessary at all.
(2) Once a reduced problem has been obtained by RSR and a candidate list is generated, we apply the approximate degenerate decoding rule $\delta$-ADD to select the most likely logical coset from the OSD list.
  
}

\section{Conclusion} \label{sec:conclusion}

{
We proposed an RSR procedure that reduces the decoding problem size using BP-generated statistics with both hard- and soft-decision reliabilities.
RSR serves as a general-purpose preprocessor that enables post-processing decoders, including OSD, LSD, and AC, to operate efficiently on the reduced problems.

We further introduced a degeneracy-aware pruning condition to identify when higher-order OSD is beneficial.
Combined with MBP, this yields the MBP--RSR--ADOSD family of decoders, which achieves strong and consistent performance across a broad range of sparse quantum codes under the code-capacity noise model, including both CSS and non-CSS codes, as well as under circuit-level DEM noise for rotated surface codes.

A key practical result is that MBP--RSR--ADOSD$_4$ scales to decoding problems with more than $10^4$ error variables.
At low physical error rates (e.g., below $10^{-3}$), errors are typically sparse and isolated, so RSR produces a small effective problem size, making higher-order OSD both computationally feasible and highly effective in this regime.

More general circuit-level decoding problems with strongly correlated syndrome structures~\cite{KL24b} are left for future work.
}

 {
  The source files of our simulations can be found in \cite{Kung2026MBP_ADOSD}.
}

\bibliographystyle{IEEEtran}
\bibliography{myBib_v2.bib}

\end{document}